\documentclass[11pt,a4paper]{article}
\usepackage{comment,listings,framed,graphicx,url,xcolor,enumerate, hyperref, amsmath, float, multicol, multirow}
\usepackage{fullpage} 
\usepackage{todonotes}
\definecolor{mild}{rgb}{1,0.98,0.8}
\usepackage[utf8]{inputenc}    
\usepackage[T1]{fontenc}       
\usepackage{lmodern}           
\usepackage[margin=.6in]{geometry}
\usepackage{graphicx}          
\usepackage{amsmath,amssymb}   
\usepackage{booktabs}          
\usepackage{caption}           
\usepackage{natbib}
\usepackage{subcaption}        
\usepackage{microtype}         
\usepackage{hyperref}          
\usepackage{natbib}            
\usepackage{xcolor}            
\usepackage{amsthm}
\usepackage{amsmath}

\numberwithin{equation}{section}
\numberwithin{figure}{section}
\numberwithin{table}{section}

\newtheorem{theorem}{Theorem}[section]   

\theoremstyle{definition}

\theoremstyle{remark}
\newtheorem*{remark}{Remark}             

\hypersetup{
  colorlinks=true,
  linkcolor=blue,
  citecolor=black,
  urlcolor=blue
}

\title{Analysis of Nonlinear Random Polarization in Dispersive Dielectrics}
\title{Analysis of Nonlinear Random Polarization in Dispersive Dielectrics}

\author{%
Nathan L.~Gibson\thanks{CONTACT Nathan L.~Gibson. Email: gibsonn@math.oregonstate.edu}
\quad 
Emmanuel E.~Oguadimma\\
Oregon State University, Corvallis, OR, USA
}

\date{\today}

\begin{document}

\maketitle

\begin{abstract}
We present a study on the time-domain propagation of electromagnetic waves in dielectric materials modeled by a nonlinear Debye medium with random perturbations. Polynomial Chaos Expansions are employed to transform the random nonlinear Debye polarization model into a deterministic framework. We extend the Yee discretization to the resulting coupled system, establish second order accuracy, and verify convergence numerically. We investigate the sensitivity of nonlinear properties to uncertainty, particularly when the amplitude of the input signal is large. Given the challenges in manufacturing where uncertainties can cause optimal parameters to vary and potentially disrupt nonlinear effects, our approach incorporates these uncertainties within the simulation. This can enable the model-based design identification of realizable materials that maintain their desired effects despite variations. The findings from this study contribute to a deeper understanding of wave propagation in complex media, with potential implications for applications in optical communications, material science, and electromagnetic wave control.
\end{abstract} 


\section{Introduction}
Ultrashort electromagnetic pulses in the terahertz (THz) frequency range have gained significant interest for their ability to probe materials with high spatial and temporal resolution. A single-cycle pulse lasting just a few picoseconds enables full-spectrum THz measurements, facilitating rapid and non-invasive characterization of material properties. In this frequency regime, several organic materials exhibit strong dispersion and absorption that reflect their underlying molecular structure \cite{HTBanksGabriella}. These features can be exploited in applications ranging from security screening and industrial inspection to medical diagnostics, including the differentiation of cancerous and benign tissue based on electromagnetic signatures.

The propagation of ultrashort electromagnetic pulses in dielectric media is governed by Maxwell’s equations coupled to constitutive relations that describe the response of the material through polarization \cite{jackson1998classical,landau1984electrodynamics}. In dispersive dielectrics, polarization exhibits temporal memory and depends on the past history of the electric field; equivalently, the medium is characterized by a frequency-dependent permittivity, which leads to pulse distortion and dispersion \cite{landau1984electrodynamics,taflove2005computational,oughstun1994pulse}. Debye-type materials provide a classical model for such relaxation behavior through a single characteristic time scale, which can strongly influence transient phenomena such as precursors and early-time waveform features \cite{pleshko1969precursors,jackson1998classical}. Under sufficiently high field intensities, additional nonlinear polarization mechanisms become relevant, motivating nonlinearly forced Debye models in which the polarization depends nonlinearly on the electric field amplitude \cite{albanese1993ultrashort,boyd2008nonlinear}. Although these nonlinear models provide important insight, they are typically formulated in deterministic settings, using fixed material parameters. In practice, however, many materials exhibit heterogeneities at the microscale that lead to parametric uncertainty in macroscopic models. For example, spatial variability in the relaxation time $\tau$ can strongly affect the dielectric response of a material~\cite{gibson2014polynomial}. These considerations motivate an uncertainty-quantification viewpoint in which uncertain constitutive parameters are represented as random variables or random fields, and their impact on quantities of interest is propagated through the governing equations using efficient stochastic methods. This type of uncertainty is not unique to electromagnetics, it also arises in other physical systems such as flow through porous media, where spatially dependent soil properties introduce variability into governing parameters~\cite{ghanem1998stochastic}.

A major challenge in modeling pulse propagation in realistic dielectric media is that material parameters may vary across samples, spatial locations, or operating conditions, so any single fixed parameter choice can fail to represent the inherent variability. A common remedy is Monte Carlo sampling, in which one repeatedly solves the governing equations for randomized parameter realizations and estimates statistics from the resulting ensemble; however, because each sample requires a full forward solve and convergence in sampling-based estimators is typically slow \cite{smith2024uncertainty, xiu2010numerical}, this approach can become prohibitively expensive for time-dependent Maxwell systems. An alternative is to treat uncertain parameters as random variables on an underlying probability space and propagate that uncertainty through the equations directly. Polynomial Chaos Expansion (PCE) provides a spectral framework for doing so by expressing the stochastic solution as a series of orthogonal polynomials in the random inputs, an idea originating in Wiener’s chaos expansion for Gaussian variables and later generalized to broader distributions through the generalized polynomial chaos (gPC) framework developed by Xiu and Karniadakis \cite{wiener1938homogeneous,xiu2002wiener}.

In this work, we incorporate uncertainty into the Debye relaxation time and study how it interacts with the cubic nonlinear polarization response in dispersive dielectric media, extending the work done in \cite{HTBanksGabriella, zaidan2010mathematical, long2011nonlinear} . We use a polynomial chaos expansion to convert the resulting random polarization model into a deterministic surrogate system that enables efficient forward simulation. To solve this coupled deterministic system, we extend the classical second-order Yee leapfrog scheme, establish $\mathcal{O}(\Delta z^2 + \Delta t^2)$ accuracy, and verify the second-order convergence rates numerically. A frequency-domain analysis and the forward simulations show that the predicted dispersion, attenuation, and time-domain waveform can change noticeably depending on whether relaxation-time variability is included and whether the nonlinear forcing is retained, particularly across the GHz and THz regimes and for larger field amplitudes. This motivates the constrained inverse design problem considered here, in which we tune the material parameter set to increase the peak electric field at a prescribed observation location while controlling nonphysical oscillations. The optimization is solved using MATLAB \texttt{fmincon} with the interior-point algorithm. We included a significance test within the optimization study to quantify the contribution of the randomness parameter $\tau_r$ and the nonlinearity parameter $\beta$, individually and in combination, to the resulting optimization outcomes.

The remainder of this paper is organized as follows. In Section~\ref{sec:model}, we present the governing equations and the stochastic formulation of the nonlinear Debye model.
In Section~\ref{LF}, we describe the numerical discretization and implementation details. Section~\ref{sec:pce} introduces the intrusive PCE framework, derives the coupled deterministic system for the PCE modes, and presents a second-order accuracy analysis of the Yee leapfrog discretization. Section~\ref{sec:results} verifies the convergence rates and reports the forward uncertainty-propagation experiments. Section~\ref{sec:parameterIdentif} presents the parameter identification study. Finally, Section~\ref{sec:conclusions} concludes with a summary and directions for future work.

\section{Maxwell's Equations}\label{sec:model} 
We present Maxwell’s equations governing the electric field $\mathbf{E}$ and the magnetic field $\mathbf{H}$ in a spatial domain $\Omega \subseteq \mathbb{R}^3$ over a finite time interval $(0,T)$, where $T>0$. We assume the absence of free charges, i.e., $\rho=0$ in $\Omega$, and impose perfectly electric conducting boundary conditions on $\partial\Omega$, so that the tangential component of $\mathbf{E}$ vanishes. With homogeneous initial data, the electromagnetic fields satisfy the following system on $(0,T)\times\Omega$:
\begin{subequations}\label{eq: 2.1}
\begin{align}
\frac{\partial \mathbf{B}}{\partial t} + \nabla\times \mathbf{E} &= 0, \quad \text{in }(0,T)\times \Omega, \label{eq:2.1a}\\
\frac{\partial \mathbf{D}}{\partial t} + \mathbf{J} - \nabla\times \mathbf{H} &= 0, \quad \text{in }(0,T)\times \Omega, \label{eq:2.1b}\\
\nabla\cdot \mathbf{B} &= 0, \quad \text{in }(0,T)\times \Omega, \label{eq:2.1c}\\
\nabla\cdot \mathbf{D} &= \rho,\quad \text{in }(0,T)\times \Omega, \label{eq:2.1d}\\
\mathbf{E} \times \mathbf{n} &= 0, \quad \text{on }(0,T)\times \partial\Omega, \label{eq:2.1e}\\
\mathbf{E}(0,\mathbf{x}) &=0,\quad \mathbf{H}(0,\mathbf{x})=0 \quad \text{in } \Omega. \label{eq:2.1f}
\end{align}
\end{subequations}
 The Constitutive Laws describe the response of the medium to the electromagnetic field, including effects from polarization. They typically take the form:
\begin{subequations}\label{eq:2.2}
\begin{align}
\mathbf{D} &= \varepsilon\mathbf{E} + \mathbf{P}, \label{eq:2.2a}\\
\mathbf{B} &= \mu\,\mathbf{H} + \mathbf{M}, \label{eq:2.2b}\\
\mathbf{J} &= \sigma\,\mathbf{E} + \mathbf{J}_s, \label{eq:2.2c}
\end{align}
\end{subequations}
where \(\varepsilon=\varepsilon_0\varepsilon_\infty\).  In \eqref{eq:2.2}, \(\mathbf{D}\) and \(\mathbf{E}\) represent the electric flux density and the electric field, respectively, \(\mathbf{P}\) refers to the macroscopic polarization, \(\mathbf{M}\) the magnetization, and \(\mathbf{J}_s\) the source current density.  The dielectric parameters are \(\varepsilon_0\), the electric permittivity of free space; \(\varepsilon_\infty\), the relative electric permittivity in the limit of infinite frequencies; \(\mu\), the magnetic permeability; and \(\sigma\), the electric conductivity.

Combining equations \eqref{eq:2.1a} and \eqref{eq:2.2b} under the assumptions that there is no magnetization and that \(\mu=\mu_0\) (i.e., free‐space permeability) gives
\begin{equation}\label{eq:2.3}
\mu_0\,\frac{\partial \mathbf{H}}{\partial t} \;=\; -\,\nabla\times \mathbf{E}.
\end{equation}
Furthermore, we combine equations \eqref{eq:2.1b}, \eqref{eq:2.2a} and \eqref{eq:2.2c} to obtain 
\begin{equation}\label{eq:2.4}
\varepsilon_0\,\varepsilon_\infty\frac{\partial \mathbf{E}}{\partial t}
=-\sigma\,\mathbf{E} - \mathbf{J}_s -\frac{\partial \mathbf{P}}{\partial t}+\nabla\times \mathbf{H}.
\end{equation}

\subsection{Nonlinear Polarization Model}\label{sec 2.2}
For the purposes of this paper we are concerned with the way in which electromagnetic waves propagate through a nonlinear Debye media. These media have a polarization which can be expressed in convolution form as:
\begin{equation}\label{eq:2.6}
\mathbf{P}(t,\mathbf{x})
= \tilde{g} * (\mathbf{E}(t,\mathbf{x}) + h(\mathbf{E}(t,\mathbf{x}))
= \int_{0}^{t} \tilde{g}\bigl(t - s,\,\mathbf{x};\,\mathbf{q}\bigr)\,\bigl(\mathbf{E}(s,\mathbf{x})+ h(\mathbf{E}(s,\mathbf{x}))\bigr)\,\mathrm{d}s,
\end{equation}
where \(\tilde{g}\) is the dielectric response function (DRF), representing the memory effect caused by the dielectric. In every practical example (Debye, Lorentz, etc.) DRFs depend on parameters as well as on time (and possibly space); we write \(\tilde{g} = \tilde{g}(t,x;\mathbf{q})\), where typically \(\mathbf{q}\) contains the high-frequency permittivity \(\varepsilon_\infty\), the static permittivity \(\varepsilon_s\), and the relaxation time \(\tau\), i.e., \(\mathbf{q} = \{\varepsilon_{s}, \varepsilon_{\infty}, \tau\}\). For Debye materials this DRF is \cite{gibson2014polynomial}:
\begin{equation}\label{eq:2.7}
\tilde{g}(t,x; \mathbf{q})
= \frac{\varepsilon_{0}(\varepsilon_{s} - \varepsilon_{\infty})}{\tau}\,e^{-t/\tau},
\end{equation}
Substituting \eqref{eq:2.7} into \eqref{eq:2.6} and differentiating yields 
\begin{equation} \label{eq: 2.8}
     \frac{\partial \mathbf{P}}{\partial t} = \frac{\varepsilon_0 (\varepsilon_{s} - \varepsilon_{\infty})}{\tau} \,  \frac{\partial}{\partial t}\, \left(\int_0^t e^{-(t-s)/\tau} \bigl(\mathbf{E}(s,\mathbf{x})+h(\mathbf{E}(s,\mathbf{x})\bigl) \right)\, ds
\end{equation}
The preceding equation further simplifies to 
\begin{align*}
    \frac{\partial \mathbf{P}}{\partial t} &= \frac{\varepsilon_0 (\varepsilon_{s} - \varepsilon_{\infty})}{\tau} \left( \bigl(\mathbf{E}(t,\mathbf{x})+h(\mathbf{E}(t,\mathbf{x}))\bigr) -\frac{1}{\tau} \int_0^t e^{-(t-s)/\tau} \bigl(\mathbf{E}(s,\mathbf{x})+h(\mathbf{E}(s,\mathbf{x}))\bigr) ds\right) \\
    &= \frac{1}{\tau} \left(\varepsilon_0 (\varepsilon_{s} - \varepsilon_{\infty}) \bigl(\mathbf{E}(t,\mathbf{x})+h(\mathbf{E}(t,\mathbf{x}))\bigr) - \mathbf{P}\right)
\end{align*}
or equivalently, 
\begin{equation}\label{eq:2.9}
     \tau \frac{\partial \mathbf{P}}{\partial t} + \mathbf{P} = \varepsilon_0 (\varepsilon_{s} - \varepsilon_{\infty}) \bigl(\mathbf{E}(t,\mathbf{x})+h(\mathbf{E}(t,\mathbf{x}))\bigr)
\end{equation}
which is a nonlinearly forced Debye polarization model. By substituting \eqref{eq:2.9} into \eqref{eq:2.4}, and combining the the resulting equations with \eqref{eq:2.3}, we obtain a \textit{nonlinearly forced Maxwell--Debye system} given by the form
\begin{subequations}\label{MD}
\begin{align}
\mu_0\,\frac{\partial \mathbf{H}}{\partial t}
&= -\,\nabla\times \mathbf{E}, \\[4pt]
\varepsilon_0\,\varepsilon_\infty\frac{\partial \mathbf{E}}{\partial t}
&=-\sigma\,\mathbf{E} - \mathbf{J}_s -\frac{\partial \mathbf{P}}{\partial t}+\nabla\times \mathbf{H}, \\[4pt]
\tau\,\frac{\partial \mathbf{P}}{\partial t} + \mathbf{P}
&= \varepsilon_0 \varepsilon_d\left(\mathbf{E}+ \beta\,|\mathbf{E}|^2\,\mathbf{E}\right),
\end{align}
\end{subequations}
This model implicitly assumes that single parameter values are representatives of the dielectric response of materials. More realistically, these parameters should be modeled as random variables with probability distributions, which we explore in the following section. 

\subsection{Random Polarization} 
In earlier studies, numerical methods for approximating solutions to Maxwell’s equations often assumed that material parameters, such as relaxation time, could be adequately represented by their mean values. Subsequent approaches improved upon this by modeling these parameters using uniform distributions. However, such approximations remain insufficient, as empirical studies suggest that relaxation times are more accurately characterized by a log-normal distribution~\cite{boettcher1974polarization}, which can be better approximated using beta or gamma distributions. To account for uncertainty in polarization mechanisms, we introduce a distribution over the relaxation parameters and refer to such materials as \emph{polydispersive}. Thus, we define our polarization model in terms of a distribution‐dependent dielectric response function $g$,
\begin{equation}\label{eq:3.1}
\mathbf{P}(t,\mathbf{x};F)
=\int_{0}^{t}g\bigl(t-s,\mathbf{x};F\bigr)\,\bigl(\mathbf{E}(s,\mathbf{x})+h(\mathbf{E}(s,\mathbf{x}))\bigr)\,ds,
\end{equation}
where $g$ is determined by a family of polarization laws, each described by a different parameter $\tau$, and therefore is given by
\begin{equation}\label{eq:2.10}
g(t,\mathbf{x};F)
=\int_{\Gamma}\tilde g(t,\mathbf{x};\tau)\,dF(\tau),
\quad \Gamma =[\tau_a,\tau_b]\subset(0,\infty).
\end{equation}
where $F(\tau)$ is the density function for $\tau$. In particular, if the distribution $F$ were discrete, consisting of a single relaxation parameter, then we would again have \eqref{eq:2.9}.  In the case where $\tau$ has a uniform distribution, $dF(\tau)=f(\tau)\,d\tau$ with $f(\tau)=1/(\tau_b-\tau_a)$.

By plugging \eqref{eq:2.10} into \eqref{eq:3.1}, the macroscopic electric polarization reads 
\[
\mathbf{P}(t,\mathbf{x})
=\int_{0}^{t}\Bigl[\int_{\Gamma}\tilde{g}\bigl(t-s,\mathbf{x};\tau\bigr)\,dF(\tau)\Bigr]\,\bigl(\mathbf{E}(s,\mathbf{x})+h(\mathbf{E}(s,\mathbf{x}))\bigr)\,ds,
\]
or equivalently, by interchanging integrals, we obtain
\begin{equation}\label{eq:3.2}
\mathbf{P}(t,\mathbf{x})
=\int_{\Gamma}\mathcal{P}(t,\mathbf{x};\tau)\,dF(\tau)
=\mathbb{E}[\mathcal{P}],
\end{equation}
where
\[
\mathcal{P}(t,\mathbf{x};\tau)
=\int_{0}^{t}\tilde{g}\bigl(t-s,\mathbf{x};\tau\bigr)\,\bigl(\mathbf{E}(s,\mathbf{x})+h(\mathbf{E}(s,\mathbf{x}))\bigr)\,ds
\]
which can be expressed as the solution to a first‐order random ordinary differential equation (RODE) where the relaxation time $\tau$ is modeled as a random variable
\begin{equation}\label{eq:2.12}
\tau\frac{\partial\mathcal{P}}{\partial t} + \mathcal{P}
= \varepsilon_0\varepsilon_d\bigl(\mathbf{E}(t,\mathbf{x})+h(\mathbf{E}(t,\mathbf{x}))\bigr).
\end{equation}
where $\varepsilon_d = \varepsilon_s - \varepsilon_\infty$. Combining \eqref{eq:2.12}, \eqref{eq:2.3}, and \eqref{eq:2.4} yields the nonlinearly forced Maxwell–random Debye system, where \eqref{eq:2.12} depends on the expected value of the random polarization, $\mathcal{P}$. 

\subsubsection{Complex Permittivity}
In the frequency domain, the linear dielectric response of a material is conveniently summarized by its \emph{complex permittivity} $\epsilon(\omega) = \epsilon_R(\omega) - i \epsilon_I(\omega)$, which encodes both dispersion and electromagnetic energy dissipation (dielectric loss) in the medium \cite{jackson1998classical,landau1984electrodynamics}. When an oscillating electric field is applied, bound charges and dipoles do not generally respond instantaneously; instead, polarization lags the driving field due to microscopic relaxation processes that depend on the material’s structure and the excitation frequency. This phase lag manifests as a frequency-dependent real part $\epsilon_R(\omega)$ (associated with energy storage) together with a nonzero imaginary part $\epsilon_I(\omega)$ (associated with energy loss) \cite{landau1984electrodynamics}. Therefore, complex permittivity is essential for modeling dispersive and lossy dielectrics, such as biological tissues, water, ceramics, and polymers~\cite{gabriel1996dielectric}. In these cases, simple real-valued models of permittivity are insufficient to accurately predict field behavior, particularly at GHz and THz frequencies. 

We next derive an expression for the complex permittivity associated with the one spatial dimensional nonlinearly forced Debye polarization model. Throughout this derivation we assume $\beta>0$ and adopt the cubic nonlinear forcing:
\begin{align}\label{eqn 2.23}
    h(E)=\beta E^3.
\end{align}
To obtain a frequency-domain representation, we take the Fourier transform in time of the polarization equation \eqref{eq:2.2a}. Denoting the Fourier transform of a time-dependent quantity $f(t)$ by $\widehat{f}(\omega)$, we obtain 
\begin{equation}\label{3.2}
    \hat{D} = \varepsilon_0 \varepsilon_\infty \hat{E} + \hat{\tilde{g}} \hat{E} + \beta \hat{\tilde{g}} \, \widehat{E^3},
\end{equation}
A direct computation of this convolution yields 
\begin{align}
    \hat{D} &= \varepsilon_0 \varepsilon_\infty \hat{E} + \frac{\varepsilon_0 \varepsilon_d}{1+i\omega \tau} \hat{E} \notag\\
    &\qquad+ \frac{\beta \varepsilon_0 \varepsilon_d}{1+i\omega \tau} \frac{1}{4\pi^2} \int_{-\infty}^\infty  \int_{-\infty}^\infty \hat{E}(\omega'')\hat{E}(\omega'-\omega'')\hat{E}(\omega - \omega')\,d\omega''\,d\omega'
\end{align}
Due to the complexity and computational cost of evaluating this term explicitly, we introduce a simplifying approximation by linearizing the nonlinear function \eqref{eqn 2.23} giving
\[
    h(E) \approx \beta |E|^2 E,
\]
where $|E|^2$ is treated as a constant envelope (in this paper, $|E|^2 = \max_{t} |E|^2$). This allows us to factor the nonlinearity out of the convolution \eqref{3.2}, significantly simplifying the analysis. Instead of using \eqref{eq:2.2a}, we consider a more general frequency-domain representation of Maxwell’s constitutive law in the presence of conductivity $\sigma \neq 0$, starting from 
\begin{align}
    \mathcal{F}\left(\frac{\partial D}{\partial t} + J\right)   &= \mathcal{F}(\dot{D} + J) \label{constlawsigma}\\
    &= \mathcal{F}(\varepsilon_0 \varepsilon_\infty \dot{E} + \dot{P} + \sigma E)\notag\\
    &= \varepsilon_0 \varepsilon_\infty\mathcal{F}(\dot{E}) + \mathcal{F}(\dot{P}) + \sigma\mathcal{F}({E})\notag\\
    &\approx \varepsilon_0 \varepsilon_\infty \mathcal{F}(\dot{E}) + \mathcal{F}(\tilde{g} * \dot{E} + \tilde{g} * \beta |E|^2 \dot{E}) + \sigma \mathcal{F}(E)   \notag\\
    &= i \omega \varepsilon_0 \varepsilon_\infty \hat{E} + i\omega\hat{\tilde{g}}\hat{E} + i \omega \hat{\tilde{g}} \beta |E|^2 \hat{E} + \sigma \hat{E} \notag\\
    &= i\omega \left(\varepsilon_0 \varepsilon_\infty \hat{E} + \frac{\varepsilon_0 \varepsilon_d}{1+i\omega\tau} \hat{E} + \frac{\varepsilon_0 \varepsilon_d \beta}{1+i\omega \tau}|E|^2 \hat{E} + \frac{\sigma}{i\omega} \hat{E}\right)\notag\\
    &= i\omega\varepsilon_0 \left(\varepsilon_\infty + \frac{\varepsilon_d}{1+i\omega\tau} + \frac{\varepsilon_d \beta |E|^2}{1+i\omega\tau} + \frac{\sigma}{i\omega \varepsilon_0}\right) \hat{E}\notag\\
    &= i\omega\varepsilon_0 \left(\varepsilon_\infty + \frac{\varepsilon_d}{1+(\omega \tau)^2} + \frac{\beta \varepsilon_d |E|^2}{1+(\omega \tau)^2} - i \left(\frac{\varepsilon_d \omega \tau}{1+(\omega \tau)^2} + \frac{\beta \varepsilon_d |E|^2 \omega \tau}{1+(\omega \tau)^2} + \frac{\sigma}{\omega \varepsilon_0}\right)\right) \hat{E}\notag\\
    &= i\omega \varepsilon_0 \varepsilon(\omega) \hat{E} \notag
\end{align}
and so, we have the complex permittivity to be given by
\begin{equation}\label{eq: 2.26}
    \varepsilon(\omega) = \varepsilon_\infty + \frac{\varepsilon_d}{1+(\omega \tau)^2} + \frac{\beta \varepsilon_d |E|^2}{1+(\omega \tau)^2} - i \left(\frac{\varepsilon_d \omega \tau}{1+(\omega \tau)^2} + \frac{\beta \varepsilon_d |E|^2 \omega \tau}{1+(\omega \tau)^2} + \frac{\sigma}{\omega \varepsilon_0}\right)
\end{equation}
with the real (permittivity) and imaginary (conductivity) parts given by
\begin{align} \label{eq: 2.27}
    \varepsilon_R(\omega) &= \varepsilon_\infty + \frac{\varepsilon_d}{1+(\omega \tau)^2} + \frac{ \varepsilon_d \beta |E|^2}{1+(\omega \tau)^2}
\end{align}
and
\begin{align}\label{eq: 2.28}
    \varepsilon_I(\omega) = \sigma(\omega) = \omega \varepsilon_0\left(\frac{\varepsilon_d \omega \tau}{1+(\omega \tau)^2} + \frac{\varepsilon_d \beta |E|^2 \omega \tau}{1+(\omega \tau)^2} + \frac{\sigma}{\omega \varepsilon_0}\right)
\end{align}
Observe that $\text{$\sigma(\omega) \to \sigma$ as } \omega \to 0$.

\subsubsection{Expected Complex Permittivity}
In order to account for uncertainty in the relaxation behavior of the material, we model the relaxation time $\tau$ as a uniformly distributed random variable. In this case, \eqref{eq: 2.26} will be referred to as a \emph{stochastic complex permittivity} \cite{bela2010generalized}. In this case, equation \eqref{constlawsigma} depends on the expected value of $D$.  Our goal is to compute the \emph{expected complex permittivity} by averaging over all possible values of $\tau$, thereby obtaining a frequency-dependent response that captures both dispersion and uncertainty. Let us assume $\xi \sim \mathcal{U}[-1,1]$. Then $\tau \sim \mathcal{U}(\tau_{a}, \tau_b)$. The probability density function for $\tau$ is given by 
\begin{align}
    f(\tau) = \left\{
\begin{array}{ll}
1/(\tau_{b} - \tau_{a}) & \tau_{a} < \tau < \tau_{b}, \\
0 & \text{otherwise}.
\end{array}
\right.
\end{align}
For any complex number $z = x + i y \in \mathbb{C}$, we know that 
\[
    \ln(z) = \ln|z| + i \arg(z), \quad \arg(z) = \tan^{-1}(y/x). 
\]
Hence, the expected complex permittivity (for $\sigma \neq 0$) is given by
\begin{align*}
    \mathbb{E}[\varepsilon(\omega)] &= \mathbb{E}\left[\varepsilon_\infty + \frac{\varepsilon_d}{1+(\omega \tau)^2} + \frac{\beta \varepsilon_d |E|^2}{1+(\omega \tau)^2} - i \left(\frac{\varepsilon_d \omega \tau}{1+(\omega \tau)^2} + \frac{\beta \varepsilon_d |E|^2 \omega \tau}{1+(\omega \tau)^2} + \frac{\sigma}{\omega \varepsilon_0}\right)\right] \\
    &= \int_{\tau_a}^{\tau_b} \varepsilon_\infty f(\tau) \, d\tau + \int_{\tau_a}^{\tau_b} \frac{\varepsilon_d}{1+(\omega \tau)^2} f(\tau) \, d\tau + \int_{\tau_a}^{\tau_b} \frac{\beta \varepsilon_d |E|^2 }{1+(\omega \tau)^2} f(\tau) \, d\tau \\
    &\quad - i \left(\int_{\tau_a}^{\tau_b} \frac{\varepsilon_d \omega \tau}{1+(\omega \tau)^2} f(\tau) \, d\tau + \int_{\tau_a}^{\tau_b} \frac{\beta \varepsilon_d |E|^2 \omega \tau}{1+(\omega \tau)^2} f(\tau) \, d\tau + \int_{\tau_a}^{\tau_b} \frac{\sigma}{\omega \varepsilon_0} f(\tau) \, d\tau \right) \\
    &= \varepsilon_\infty + \frac{\varepsilon_d}{\omega(\tau_b - \tau_a)} \tan^{-1}(\omega \tau) \bigg|_{\tau_a}^{\tau_b} + \frac{\beta |E|^2 \varepsilon_d}{\omega(\tau_b - \tau_a)} \tan^{-1}(\omega \tau) \bigg|_{\tau_a}^{\tau_b} \\
    &\quad - i \left( \frac{\varepsilon_d}{2\omega(\tau_b - \tau_a)} \ln(1+(\omega \tau)^2) \bigg|_{\tau_a}^{\tau_b} + \frac{\beta \varepsilon_d |E|^2}{2\omega(\tau_b - \tau_a)} \ln(1+(\omega \tau)^2) \bigg|_{\tau_a}^{\tau_b} + \frac{\sigma}{\omega \varepsilon_0} \right)
\end{align*}
with the real (permittivity) and imaginary (conductivity) parts given by 
\begin{align}\label{eq: 2.30}
    \varepsilon_R(\omega) = \varepsilon_\infty + \frac{\varepsilon_d}{\omega(\tau_b - \tau_a)} \tan^{-1}(\omega \tau) \bigg|_{\tau_a}^{\tau_b} + \frac{\beta |E|^2 \varepsilon_d}{\omega(\tau_b - \tau_a)} \tan^{-1}(\omega \tau) \bigg|_{\tau_a}^{\tau_b}
\end{align}
\begin{align}\label{eq: 2.31}
     \varepsilon_I(\omega) = \omega \varepsilon_0 \left( \frac{\varepsilon_d}{2\omega(\tau_b - \tau_a)} \ln(1+(\omega \tau)^2) \bigg|_{\tau_a}^{\tau_b} + \frac{\beta \varepsilon_d |E|^2}{2\omega(\tau_b - \tau_a)} \ln(1+(\omega \tau)^2) \bigg|_{\tau_a}^{\tau_b} + \frac{\sigma}{\omega \varepsilon_0} \right)
\end{align}
respectively. For the upper and lower bounds of the $\tau$ distribution, since $\tau_a < \tau_b$, we obtain
\[
    \varepsilon_{R}^U(\omega) = \varepsilon_\infty + \frac{\varepsilon_d}{1+(\omega \tau_a)^2}+ \frac{ \varepsilon_d \beta |E|^2}{1+(\omega \tau_a)^2}, \quad \varepsilon_{R}^L(\omega) = \varepsilon_\infty + \frac{\varepsilon_d}{1+(\omega \tau_b)^2}+ \frac{ \varepsilon_d \beta |E|^2}{1+(\omega \tau_b)^2}
\]
\[
    \sigma^U(\omega) = \omega \varepsilon_0 \left( \frac{\varepsilon_d \omega \tau_a}{1+(\omega \tau_a)^2} + \frac{\varepsilon_d \beta |E|^2 \omega \tau_a}{1+(\omega \tau_a)^2} + \frac{\sigma}{ \omega \varepsilon_0 } \right),
\]
\[
    \sigma^L(\omega) = \omega \varepsilon_0 \left( \frac{\varepsilon_d \omega \tau_b}{1+(\omega \tau_b)^2} + \frac{\varepsilon_d \beta |E|^2 \omega \tau_b}{1+(\omega \tau_b)^2} + \frac{\sigma}{ \omega \varepsilon_0 } \right)
\]

\begin{figure}[H]
    \centering
    \begin{subfigure}{0.48\linewidth}
         \includegraphics[width=\linewidth]{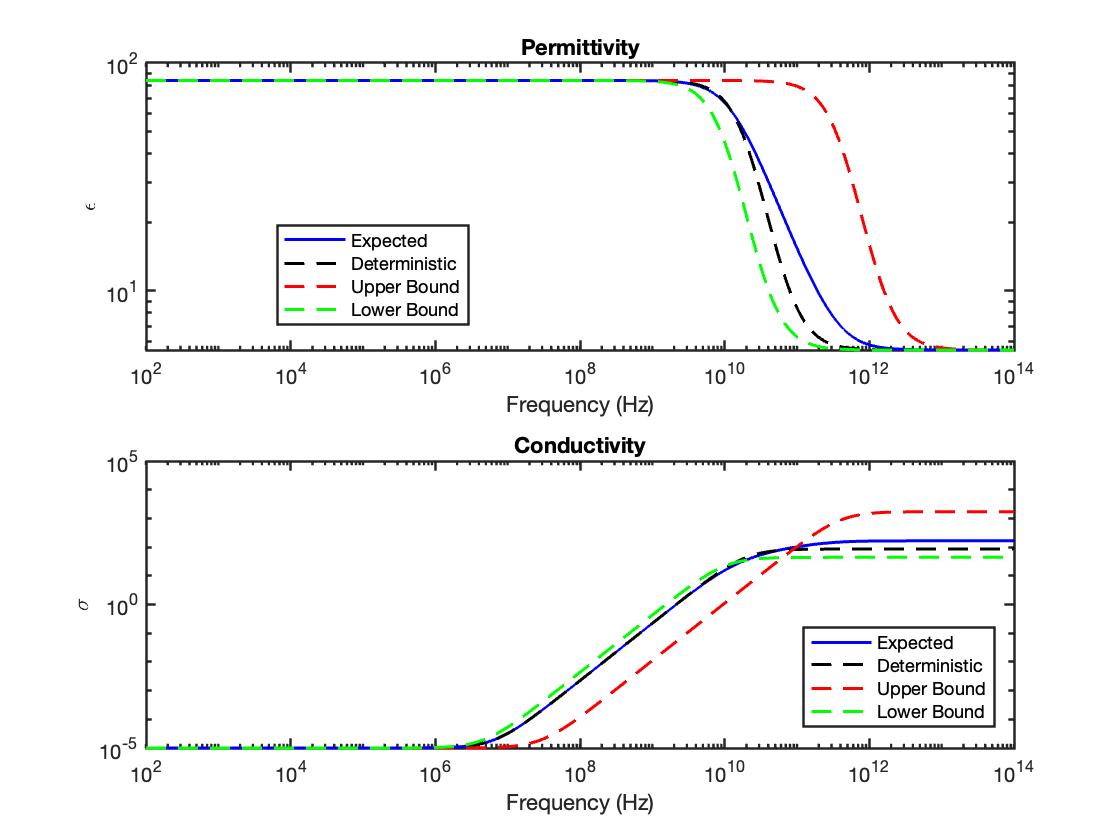}
         \caption{$|E|^2 = 10^6, \, \beta = 5 \times 10^{-6}, \, \tau_r = .95\tau_m, \, Q = 5$}
    \end{subfigure}\hfill 
    \begin{subfigure}{0.48\linewidth}
        \includegraphics[width=\linewidth]{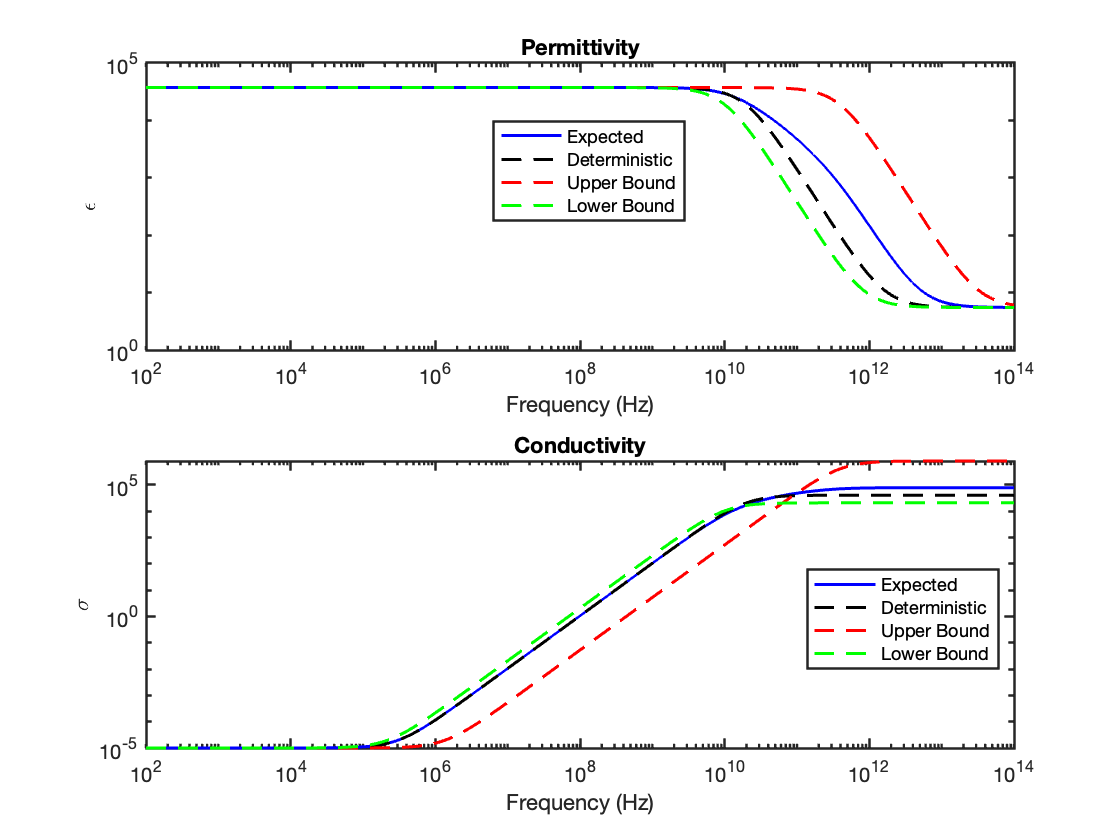}
         \caption{$|E|^2 = 10^8, \, \beta = 5 \times 10^{-6}, \, \tau_r = .95\tau_m, \, Q = 5$}
    \end{subfigure}
    \caption{Complex permittivity of the nonlinearly forced Debye model under relaxation-time uncertainty.}
    \label{fig:nonlin}
\end{figure}

\begin{figure}[H]
  \centering
  \begin{subfigure}{0.48\linewidth}
    \centering
    \includegraphics[width=\linewidth]{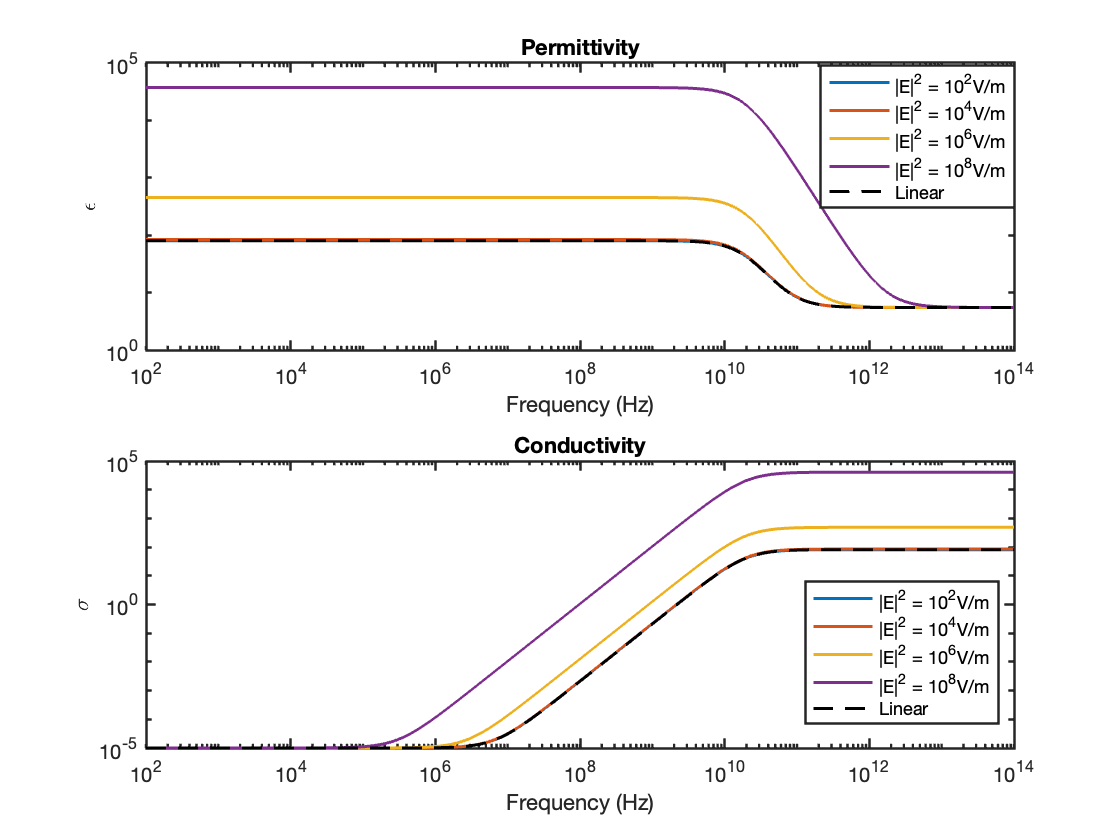}
    \caption{$\beta = 5 \times 10^{-6}, \; \tau_r = .5\tau_m, \, Q = 5$}
  \end{subfigure}\hfill
  \begin{subfigure}{0.48\linewidth}
    \centering
    \includegraphics[width=\linewidth]{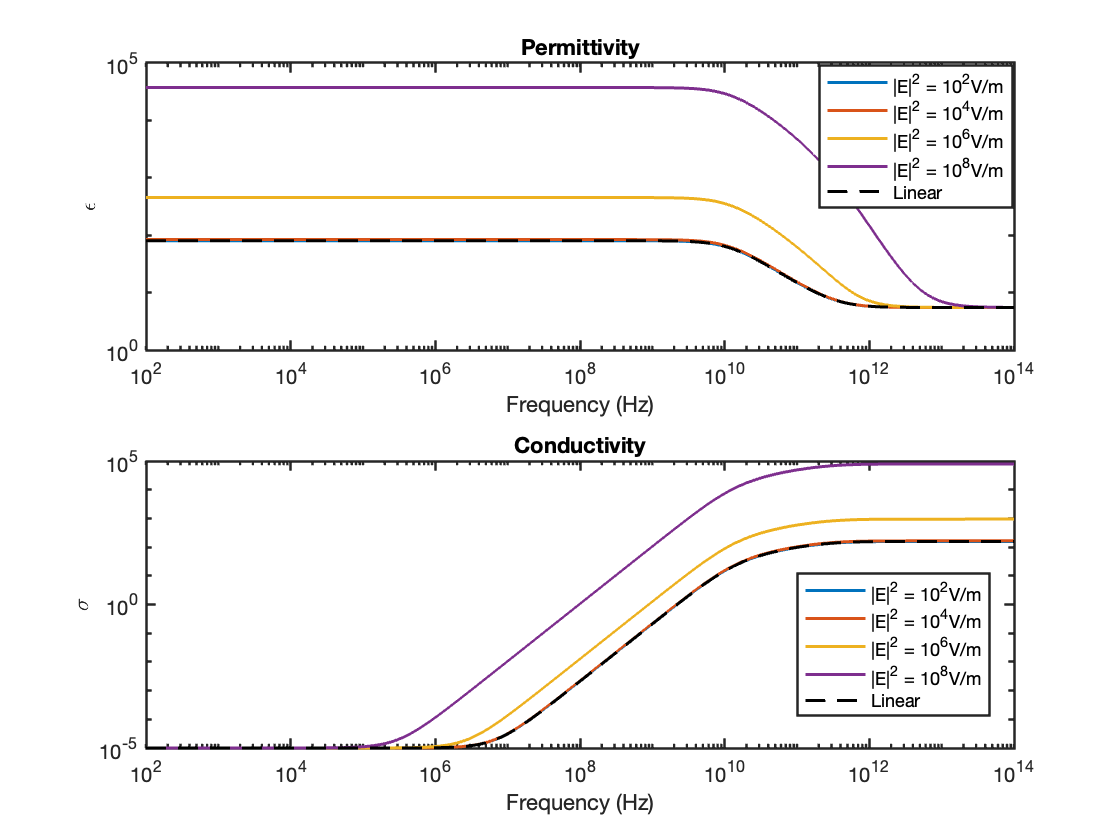}
    \caption{$\beta = 5 \times 10^{-6}, \; \tau_r = .95\tau_m, \, Q = 5$}
  \end{subfigure}
  \caption{Complex permittivity of the linear Debye model and the nonlinearly forced Debye model under relaxation-time uncertainty, shown for multiple field intensities.}
  \label{fig:nonlinran}
\end{figure}
As depicted in Fig.~\ref{fig:nonlin}, for higher frequencies ($f\geq 10 $ GHz), the gradual decrease in permittivity values for the linearized expected permittivity suggests that the electric field in the random model should propagate more slowly, while its higher conductivity indicates greater attenuation than the deterministic model. Fig.~\ref{fig:nonlinran} shows the propagation in the linear model appears to be the fastest (lowest permittivity) with the least attenuation (lowest conductivity), which we later see agrees with our time-domain simulations. However in our simulations, the electric field $E$ need not be as high as the linearized complex permittivity analysis suggests in order to observe differences.

\subsubsection{Cole-Cole Plot}
A Cole--Cole plot is a complex-plane representation of dielectric dispersion in which the imaginary part \eqref{eq: 2.28} of the complex permittivity is plotted parametrically against the real part \eqref{eq: 2.27}, with the angular frequency $\omega$ serving as the parameter \cite{cole1942dispersion,cole1941dispersion,kremer2002bds}. This visualization is widely used because it compresses broadband frequency-response information into a single geometric curve and makes deviations from ideal Debye relaxation immediately apparent \cite{woodward2021bds,barsoukov2005impedance}. For a homogeneous dielectric governed by a single Debye relaxation time, the Cole--Cole locus forms a perfect semicircular arc; by contrast, heterogeneous media in which polarization relaxes through a distribution of time scales typically produce a depressed (non-semicircular) arc whose shape is often captured phenomenologically by the Cole--Cole model \cite{cole1941dispersion,barsoukov2005impedance,kremer2002bds, gibson2014polynomial}. These observations motivate modeling relaxation behavior beyond a single deterministic time scale.

\begin{figure}[H]
  \centering
  \begin{subfigure}{0.48\linewidth}
    \centering
     \includegraphics[width=\linewidth]{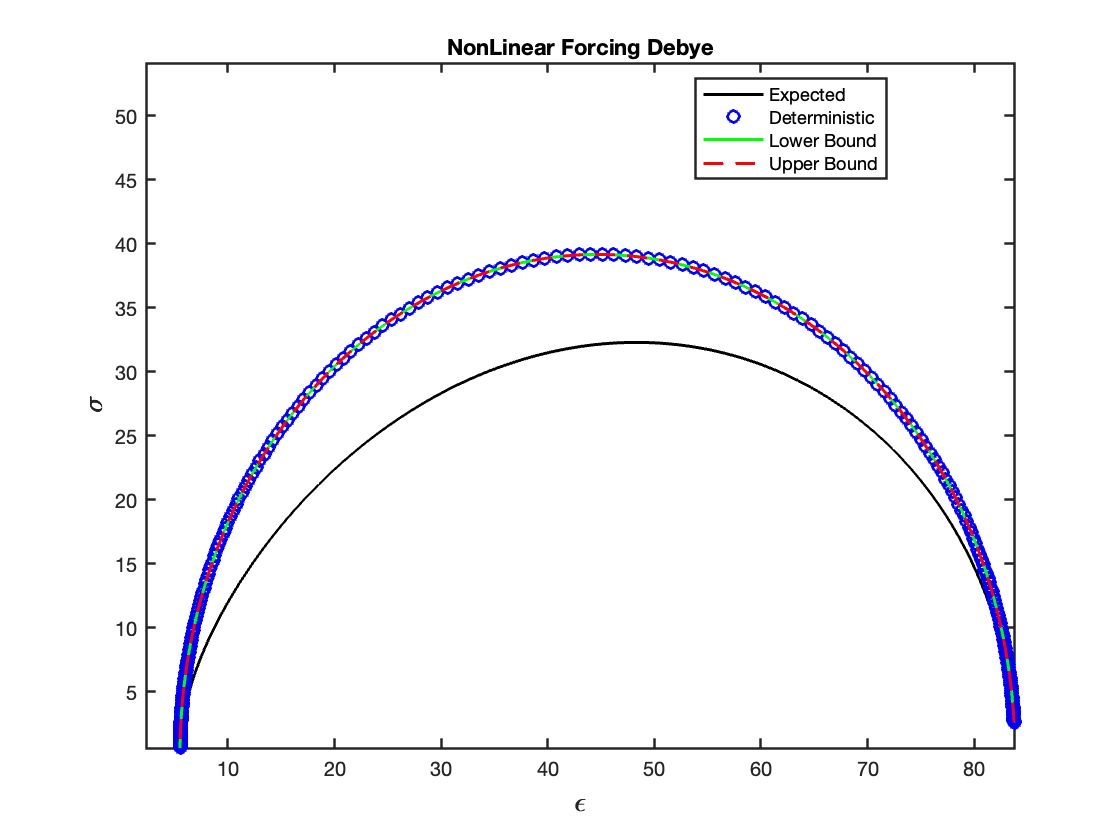}
    \caption{$|E|^2 = 10^4, \; \tau_r = .95\tau_m, \, Q = 5$ }
  \end{subfigure}\hfill
  \begin{subfigure}{0.48\linewidth}
    \centering
    \includegraphics[width=\linewidth]{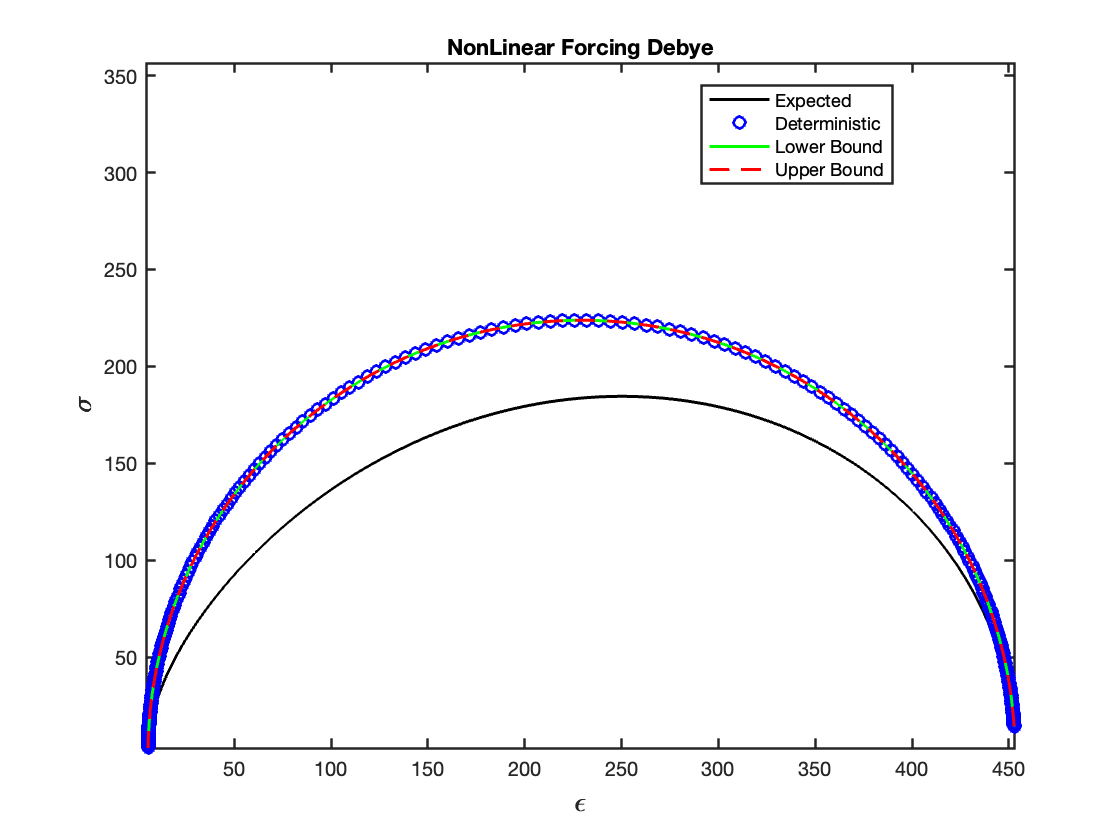}
    \caption{$|E|^2 = 10^6, \; \tau_r = .95\tau_m, \, Q = 5$  }
  \end{subfigure}


  \caption{Cole--Cole plots of the dielectric response for the nonlinearly forced Debye model under relaxation-time uncertainty.}
  \label{fig:figcol}
\end{figure}
Fig.~\ref{fig:figcol} shows that the conductivity and permittivity values increase in both the deterministic and random models with increasing field intensity, $|E|^2$. This suggests that the electromagnectic wave from the linear model has less attenuation and more speed of propagation than the nonlinear model. Since each of the Debye curves has the same shape, only shifted, these Cole-Cole plots coincide with each other. The random model, however, does not have the same shape as the deterministic Debye model and more closely resembles real material responses.

\section{Fully Discrete Leapfrog FDTD Method}\label{LF}
In this section, we construct and analyze spatial and temporal discretizations for the one dimensional form of \eqref{MD}.  Consider the spatial domain $\Omega = [0,a]\subset\mathbb{R}$ and time interval $[0,T]$ with $a,T>0$, spatial step‐size $\Delta z>0$, and time‐step $\Delta t>0$.  The discretization of $[0,a]$ and $[0,T]$ is performed as follows.  Let $L =a/\Delta z, N = T/\Delta t$. For $\ell,n\in\mathbb{N}$ we consider
\begin{subequations}
\begin{align}
0 &= z_0 \le z_1 \le \cdots \le z_\ell \le \cdots \le z_L = a,\\
0 &= t_0 \le t_1 \le \cdots \le t_n \le \cdots \le t_N = T,
\end{align}
\end{subequations}
where $z_\ell = \ell\,\Delta z, t_n = n\,\Delta t.$ Let $ (z_\alpha,t_\gamma) = \bigl(\alpha\Delta z,\;\gamma\Delta t\bigr),
    \alpha\in\{\ell,\ell+\tfrac12\},\;\gamma\in\{n,n+\tfrac12\}$.
The Yee scheme staggers the fields in space and time.  The electric field $E$ and $P$ are discretized at integer‐points in space and integer times, while the magnetic field $H$ is discretized at half-points in space and half‐integer times. We set
\begin{subequations}
\begin{align}
\Omega_h^E &:= \bigl\{z_{\ell}\mid 0\le\ell\le L\bigr\}, \\
\Omega_h^H &:= \bigl\{z_{\ell+\tfrac12}\mid 0\le\ell\le L-1\bigr\}
\end{align}
\end{subequations}
and for time levels $E,P\text{ at }t^n,\text{ for }0\le n\le N,
\, \text{and } H\text{ at }t^{n+\tfrac12},\text{ for }\;0\le n\le N-1.$ Let $U$ denote either $E$ or $P$ on $\Omega_h^E$ at $t^n$, or $H$ on $\Omega_h^H$ at $t^{n+\frac12}$.  We define the grid functions $ U^\gamma_{\alpha} \approx U\bigl(t_\gamma,\,z_\alpha\bigr),$ and write $U(t_\gamma)$ for the continuous solution, $U^\gamma$ for its grid approximation. 

Next, define the discrete $l^2$ spaces:
\begin{subequations}
\begin{align}
    V_E^h &:= \Bigl\{F:\Omega_h^E\to\mathbb{R}\;\Big|\;\|F\|_E<\infty\Bigr\},\\
    V_H^h &:= \Bigl\{U:\Omega_h^H\to\mathbb{R}\;\Big|\;\|U\|_H<\infty\Bigr\},
\end{align}
\end{subequations}
with norms
\begin{subequations}
\begin{align}
\|F\|_E^2
  &= \Delta z
    \sum_{\ell=0}^{L-1}
    \bigl|F_{\ell}\bigr|^2, \quad 
\|U\|_H^2
  = \Delta z
    \sum_{\ell=0}^{L}
    \bigl|U_{\ell+\tfrac12}\bigr|^2,
\end{align}
\end{subequations}
Our fully discrete, second-order leapfrog FDTD schemes for discretizing the nonlinearly forced Maxwell problem \eqref{MD} are given by the following.

Assume $ E_{h}^0, P_{h}^0 \in V_{E}^h$ and $H_{h}^{\frac{1}{2}} \in V_{H}^h$ are given. Then $\forall \, n \in \mathbb{N}$, with $0 < n \leq N$, find $ E_{h}^n, P_{h}^n, \in V_{E}^h$ and $H_{h}^{n+\frac{1}{2}} \in V_{H}^h$ so that $\forall \, 0 \leq j \leq J$ we have
\begin{subequations}\label{3.9}
    \begin{align}
     \mu_0 \frac{H^{n+\frac{1}{2}}_{j+\frac{1}{2}}-H^{n-\frac{1}{2}}_{j+\frac{1}{2}}}{\Delta t} &= -\frac{E^n_{j+1}-E^n_j}{\Delta z}\label{eq: 3.11a}\\
    \varepsilon_0 \varepsilon_\infty \frac{E^{n+1}_j - E^n_j}{\Delta t} = -\frac{H^{n+\frac{1}{2}}_{j+\frac{1}{2}}-H^{n+\frac{1}{2}}_{j-\frac{1}{2}}}{\Delta z}
    &- \frac{P^{n+1}_j - P^n_j}{\Delta t}
    - \sigma\,\frac{E^{n+1}_j+E^n_j}{2}-J^{n+\frac{1}{2}}_s \label{eq: 3.11b}\\
     \tau \frac{P^{n+1}_j-P^n_j}{\Delta t} + \frac{P^{n+1}_j+P^n_j}{2} &=  \varepsilon_0 \varepsilon_d\left(\frac{E^{n+1}_j + E^n_j}{2}
     + \beta \left(\frac{E^{n+1}_j+E^n_j}{2}\right)^3\right). \label{eq: 3.11c}
\end{align}
\end{subequations}
The following theorem shows that the leap-frog discretization in \eqref{3.9} is second-order accurate in both space and time. We omit the proof as it is a special case of the convergence proof for the PCE discrete system that appear in later section. 
\begin{theorem}
Suppose that the solutions to the one-dimensional Maxwell-Nonlinear Debye system \eqref{MD} are smooth enough and satisfy the regularity conditions $E, P, H \in C^3([0, T]; [C^3(\overline{\Omega})])$. Let $\xi^{n}_{H,\,j+\frac12}, \xi^{n+\frac12}_{P,\,j}, \xi^{n+\frac12}_{E,\,j}$ be the truncation errors for the Yee discretizations in \eqref{3.9}. Then for any $j\text{ and }n$, we have 
    \begin{equation}
        \max_{n,j}\left\{\bigl|\xi^{n}_{H,\,j+\frac12}\bigr|,\ \bigl|\xi^{n+\frac{1}{2}}_{P,\,j}\bigr|,\ \bigl|\xi^{n+\frac{1}{2}}_{E,\,j}\bigr|\right\}
        = \mathcal{O}(\Delta z^2+\Delta t^2).
    \end{equation} 
\end{theorem}

\section{Polynomial Chaos Expansion}\label{sec:pce} 
To efficiently solve the Maxwell-Random Debye system \eqref{eq:2.12}, we seek the expected value of the random polarization. A traditional approach applies numerical quadrature to this expectation \cite{banks2006electromagnetic}, resulting in a weighted sum of deterministic Debye models, which is equivalent to the classical multi-pole Debye approximation \cite{hurt2007multiterm}. Instead, we employ a spectral stochastic method known as Polynomial Chaos expansion (PCE) \cite{xiu2002wiener,xiu2010numerical}, which represents stochastic solutions as truncated expansions of orthogonal polynomials in the random input. Different types of orthogonal polynomials may be chosen to achieve better convergence, sometimes exponential. The chosen polynomial basis is tailored to the probability distribution of the input variables, ensuring orthogonality with respect to the associated measure. This approach transforms the stochastic system into a deterministic set of auxiliary equations for the Polynomial Chaos coefficients. We begin by recalling the three-dimensional random nonlinearly forced Debye model \eqref{eq:2.12}. For the numerical scheme developed in this section, we restrict attention to its one-dimensional reduction, given by
\begin{equation}\label{eq: 2.13}
\tau(\xi)\frac{\partial\mathcal{P}}{\partial t} + \mathcal{P}
= \epsilon_0(\epsilon_s - \epsilon_\infty)\left(E + h(E)\right), \quad \tau(\xi)=\tau_r\,\xi+\tau_m
\end{equation}
We approximate the random polarization field in the random dimension by expanding it in a truncated orthogonal polynomial basis:
\begin{equation} \label{eq: 2.14}
    \mathcal P(t,\xi) =\sum_{i=0}^Q\alpha_i(t)\,\varphi_i(\xi)
\end{equation}
where $\varphi_i(\xi)$ is the $i^{\text{th}}$ orthogonal polynomial spanning our random space and $\xi$ is a random variable with some standard distribution (e.g. $\xi \sim \mathcal{B}(a,b)$ or $\mathcal{U}[-1,1]$), and where $\tau_m$ and $\tau_r$ are shift and scaling parameters, respectively, with $\tau_m,\tau_r>0$. Plugging \eqref{eq: 2.14} into \eqref{eq: 2.13}, we obtain 
\begin{equation} \label{eq: 2.15}
(\tau_r\,\xi+\tau_m)\sum_{i=0}^Q\dot\alpha_i\,\varphi_i+\sum_{i=0}^Q\alpha_i\,\varphi_i=\epsilon_0(\epsilon_s-\epsilon_\infty)(E+h(E))
\end{equation}
All orthogonal polynomials have a recurrence relation which can be expressed in the form \cite{xiu2010numerical, lucor2004generalized}: 
\begin{equation}\label{eq: 2.16}
    \xi\,\varphi_i(\xi) = a_i\,\varphi_{i+1}(\xi)+b_i\,\varphi_i(\xi)+ c_i\,\varphi_{i-1}(\xi)
\end{equation}
valid for \(i=0,\dots,Q\) (with the convention \(\varphi_{-1}=0\)). Substituting \eqref{eq: 2.16} into \eqref{eq: 2.15} to remove the dependence on the random variable $\xi$ gives 
\begin{equation}\label{eq: 2.17}
\tau_r\sum_{i=0}^Q\dot\alpha_i\bigl(a_i\varphi_{i+1}+b_i\varphi_i+c_i\varphi_{i-1}\bigr)
+\tau_m\sum_{i=0}^Q\dot\alpha_i\,\varphi_i
+\sum_{i=0}^Q\alpha_i\,\varphi_i
=\epsilon_0(\epsilon_s-\epsilon_\infty)(E+h(E))
\end{equation}
Projecting onto the finite dimensional random space spanned by $\{\varphi_j(\xi)\}_{j=0}^Q$ by taking the weighted inner product with each basis function produces for each \(j=0,\dots,Q\),
\begin{align}
&\tau_r\sum_{i=0}^Q\dot\alpha_i\bigl(a_i\langle\varphi_{i+1},\varphi_j\rangle_w
+b_i\langle\varphi_i,\varphi_j\rangle_w
+c_i\langle\varphi_{i-1},\varphi_j\rangle_w\bigr) \notag\\
&\quad+\;\tau_m\sum_{i=0}^Q\dot\alpha_i\,\langle\varphi_i,\varphi_j\rangle_w
+\sum_{i=0}^Q\alpha_i\,\langle\varphi_i,\varphi_j\rangle_w
=\epsilon_0(\epsilon_s-\epsilon_\infty)(E+h(E))\,\langle1,\varphi_j\rangle_w.
\end{align}
In the above, $\langle \varphi_i,\varphi_j\rangle_w.$ is the weighted inner product of $\varphi_i$ and $\varphi_j$ and is defined to be: 
\begin{equation}
    \langle \varphi_i,\varphi_j\rangle_w = \int_\Gamma \varphi_i(\xi)\,\varphi_j(\xi)\,W(\xi)\,d\xi = \delta_{ij} \sqrt{h_i\, h_j}
\end{equation}
where $\delta_{ij}$ is the Kronecker delta function, $W(\xi)$ is chosen to be the weighting function with respect to which polynomials are orthogonal over the domain $\Gamma$ and $h_i = \langle\varphi_i,\varphi_i\rangle_w$. Then the \(j\)th equation becomes
\[
\tau_r\sum_{i=0}^Q\dot\alpha_i\,(a_i\,\delta_{\,i+1,j}+b_i\,\delta_{\,i,j}+c_i\,\delta_{\,i-1,j})\,h_j
+\tau_m\,\dot\alpha_j\,h_j
+\alpha_j\,h_j
=\epsilon_0(\epsilon_s-\epsilon_\infty)(E+h(E))\,\delta_{0j}\,h_0.
\]
Dividing through by \(h_j\) and letting \(\alpha=(\alpha_0,\dots,\alpha_p)^T\), we obtain the vector form, 
\begin{equation} \label{eq: 2.20}
\bigl(\tau_r\,\mathcal{M}+\tau_m\,I\bigr)\dot{\vec{\alpha}} +\vec{\alpha} = \vec{f} \implies A\dot{\vec{\alpha}} + \vec{\alpha} =  \vec{f}, 
\end{equation}
where  for \(i,j=0,\dots,Q\), the matrix \(\mathcal{M}\) is defined as $\mathcal{M}_{j,i}=a_i\,\delta_{\,j,i+1}+b_i\,\delta_{\,j,i}+c_i\,\delta_{\,j,i-1}$, and the forcing vector $\vec{f}= \epsilon_0 (\epsilon_s - \epsilon_\infty)(E + h(E))\hat{e}_1$ where $a_i$, $b_i$, and $c_i$ are the recursion coefficients. Note that the deterministic value $\vec{f}$ forces the system and is dependent on the electric field, which itself depends on the expected value 
\begin{equation}
   P = \mathbb{E}[\mathcal{P}] \approx \alpha_0.
\end{equation}
If the chosen set of orthogonal polynomials are orthogonal with respect to the density function of the random variable $\xi$, then the error in the projection, and in the mean value, converges exponentially with $Q$ \cite{bela2010generalized}. We replace the RODE \eqref{eq:2.12} with the deterministic system of ODEs \eqref{eq: 2.20}. Concretely, the \textit{nonlinearly forced Maxwell-PC Debye system} takes the form
\begin{subequations} \label{MPC}
    \begin{align}
    \mu_0\,\frac{\partial H}{\partial t} &= -\,\frac{\partial E}{\partial z}\\
    \varepsilon_0\,\varepsilon_\infty\frac{\partial E}{\partial t} =-\sigma E &- J_s -\frac{\partial \alpha_0} {\partial t}+\frac{\partial H}{\partial z} \\
    A\dot{\alpha} + \alpha &=  \epsilon_0 \epsilon_d(E + \beta E^3)\hat{e}_1. \label{4.11c}
\end{align}
\end{subequations}

\begin{remark}
For Beta random variables, the Jacobi polynomials are used, and for uniform random variables, Legendre polynomials are used~\cite{xiu2010numerical}. As an example, if we choose a uniform random variable, then for the case when $Q = 2$ we obtain the following system for the PC Debye model:
\[
A = \begin{pmatrix}
\tau_m & \frac{1}{3}\tau_r & 0 \\
\tau_r & \tau_m & \frac{2}{5}\tau_r \\
0 & \frac{2}{3}\tau_r & \tau_m
\end{pmatrix}.
\]
However, if $\xi \sim \mathcal{B}(2,5)$, and $Q=2$, then
\[
A = \begin{pmatrix}
\frac{1}{3}\tau_r + \tau_m & \frac{2}{5}\tau_r & 0 \\
\frac{2}{9}\tau_r & \frac{7}{33}\tau_r + \tau_m & \frac{14}{33}\tau_r \\
0 & \frac{18}{55}\tau_r & \frac{21}{143}\tau_r + \tau_m
\end{pmatrix}.
\]
Clearly as $\tau_r \to 0$ then the PC Debye model converges to a diagonal, uncoupled Debye ODE with a single relaxation time of $\tau = \tau_m$, which can be thought of as modeled by a delta distribution.
\end{remark}

\begin{remark}[The Invertibility of $A$]
In order to solve the PC system for $\dot{\vec\alpha}$, the matrix $A$ must be non-singular. This is true as long as $\tau_r < \tau_m$ since the eigenvalues of $\mathcal{M}$ are in $(-1,1)$~\cite{golub1969calculation}, and therefore the eigenvalues of $A$ are greater than $\tau_m - \tau_r$, thus $A$ is in fact positive definite. The reader is encouraged to see \cite{oguadimma2025analysis} for a detailed discussion on this. 
\end{remark}

Applying the discretizations discussed in the preceding section to \eqref{MPC} gives 
\begin{subequations}\label{eq: 3.21}
    \begin{align}
    \mu_0 \frac{H_{j+\frac{1}{2}}^{n+\frac{1}{2}} - H_{j+\frac{1}{2}}^{n-\frac{1}{2}}}{\Delta t} &= - \frac{E^n_{j+1}-E^n_j}{\Delta z} \label{eq: 3.21a}\\
    \varepsilon_0 \varepsilon_\infty\frac{E^{n+1}_j - E^n_j}{\Delta t} = - \frac{H_{j+\frac{1}{2}}^{n+\frac{1}{2}} - H_{j-\frac{1}{2}}^{n+\frac{1}{2}}}{\Delta z} & -\frac{\alpha_{0_j}^{n+1} - \alpha_{0_j}^{n}}{\Delta t} -\sigma \frac{E^{n+1}_j + E^n_j}{2} -J^{n+\frac{1}{2}}_s \label{eq: 3.21b}\\
    A\,\frac{\vec{\alpha_j}^{n+1} - \vec{\alpha_j}^{n}}{\Delta t} + \frac{\vec{\alpha_j}^{n+1} + \vec{\alpha_j}^{n}}{2} &= \varepsilon_0\varepsilon_d \left(\frac{E^{n+1}_j + E^n_j}{2} + \beta\left(\frac{E^{n+1}_j + E^n_j}{2}\right)^3\right)\hat{e}_1  \label{eq: 3.21c}
\end{align}
\end{subequations}
where $\alpha_0^n \in V_E$, while the grid function $\vec{\alpha}$ is taken to be from the following staggered $\ell^2$
normed space
\[
V_\alpha := \left\{ \vec{\alpha}:\tau_h^{E}\to \mathbb{R}^{Q+1}
\;\bigg|\;
\vec{\alpha}=[\alpha_0,\ldots,\alpha_Q],\, \alpha_k\in V_E,\, \|\vec{\alpha}\|_\alpha<\infty
\right\},
\]
where the discrete $L^2$ grid norm is defined as
\begin{align*}\label{PCEnorm}
\|\vec{\alpha}\|_\alpha^{2}=\sum_{k=0}^{Q}\|\alpha_k\|_{E}^{2},\qquad \forall\,\vec{\alpha}\in V_\alpha,
\end{align*}
We now show that the Yee leapfrog discretization in \eqref{eq: 3.21} is second-order accurate in both space and time.
\begin{theorem}
Suppose that the solutions to the system \eqref{MPC} are smooth enough and satisfy the regularity conditions $E,H \in C^{3}\big([0,T];[C^{3}(\Omega)]\big)$, $\alpha_k \in C^{3}\big([0,T];[C^{3}(\Omega)]\big)\ \text{for }k=0,1,\dots,Q,$ and write $\vec{\alpha}:=(\alpha_0,\alpha_1,\dots,\alpha_Q)^\top$. Let
$\xi^{n}_{H,j+\frac12}$, $\xi^{n+\frac12}_{E,j}$, and $\vec{\xi}^{\,\,n+\frac12}_{\alpha,j}$
be the truncation errors for the Yee discretizations in \eqref{eq: 3.21}. Then for any $j$ and $n$, we have
\[
\max_{n,j}\Big\{
\big|\xi^{n}_{H,j+\frac12}\big|,
\big|\xi^{n+\frac12}_{E,j}\big|,
\big\|\vec{\xi}^{\,\,n+\frac12}_{\alpha,j}\big\|
\Big\}
= \mathcal{O}(\Delta z^{2}+\Delta t^{2}).
\]

\end{theorem}

\begin{proof}
    For simplicity, we treat all physical parameters $\mu_0,\varepsilon_0,\sigma,\varepsilon_\infty,
\varepsilon_d,\beta$ and the matrix $A$ as fixed constants. Define the local truncation errors by
\[
\xi^{n}_{H,j+\frac12}
=\mu_0\frac{H\!\left(t_{n+\frac12},z_{j+\frac12}\right)-H\!\left(t_{n-\frac12},z_{j+\frac12}\right)}{\Delta t}
+\frac{E(t_n,z_{j+1})-E(t_n,z_j)}{\Delta z},
\]
\[
\begin{aligned}
\xi^{n+\frac12}_{E,j}
:=\varepsilon_0\varepsilon_\infty\frac{E(t_{n+1},z_j)-E(t_n,z_j)}{\Delta t}
&+\frac{H\!\left(t_{n+\frac12},z_{j+\frac12}\right)-H\!\left(t_{n+\frac12},z_{j-\frac12}\right)}{\Delta z}\\
&+\frac{\alpha_0(t_{n+1},z_j)-\alpha_0(t_n,z_j)}{\Delta t}
+\sigma\frac{E(t_{n+1},z_j)+E(t_n,z_j)}{2}
+J_s\!\left(t_{n+\frac12},z_j\right),
\end{aligned}
\]
and
\begin{align*}
    \vec{\xi}^{n+\frac12}_{\alpha,j}
:=A\frac{\vec{\alpha}(t_{n+1},z_j)-\vec{\alpha}(t_n,z_j)}{\Delta t}
&+\frac{\vec{\alpha}(t_{n+1},z_j)+\vec{\alpha}(t_n,z_j)}{2}\\
&-\varepsilon_0\varepsilon_d\left(
\frac{E(t_{n+1},z_j)+E(t_n,z_j)}{2}
+\beta\left(\frac{E(t_{n+1},z_j)+E(t_n,z_j)}{2}\right)^3
\right)\hat e_1.
\end{align*}
We use Taylor expansions at the natural staggered points. First, from \eqref{eq: 3.21a} we expand in time about $t_n$
and in space about $z_{j+\frac12}$, giving
\begin{align*}
    \frac{H\!\left(t_{n+\frac12},z\right)-H\!\left(t_{n-\frac12},z\right)}{\Delta t}
&=H_t(t_n,z)+\mathcal{O}(\Delta t^{2}),\\
\frac{E(t_n,z_{j+1})-E(t_n,z_j)}{\Delta z}
&=E_z(t_n,z_{j+\frac12})+\mathcal{O}(\Delta z^{2}).
\end{align*}
Therefore,
\[
\xi^{n}_{H,j+\frac12}
=\mu_0 H_t(t_n,z_{j+\frac12})+E_z(t_n,z_{j+\frac12})
+\mathcal{O}(\Delta t^{2}+\Delta z^{2}).
\]
Using the continuous Maxwell equation $\mu_0H_t+E_z=0$ evaluated at $(t_n,z_{j+\frac12})$ gives $\xi^{n}_{H,j+\frac12}=\mathcal{O}(\Delta t^{2}+\Delta z^{2}).$ Next, we expand the terms in \eqref{eq: 3.21b} at $(t_{n+\frac12},z_j)$. The standard midpoint identities yield
\begin{align*}
    \frac{E(t_{n+1},z_j)-E(t_n,z_j)}{\Delta t}
&=E_t(t_{n+\frac12},z_j)+\mathcal{O}(\Delta t^{2}),\\
\frac{\alpha_0(t_{n+1},z_j)-\alpha_0(t_n,z_j)}{\Delta t}
&=\alpha_{0_t}(t_{n+\frac12},z_j)+\mathcal{O}(\Delta t^{2}),\\
\frac{E(t_{n+1},z_j)+E(t_n,z_j)}{2}
&=E(t_{n+\frac12},z_j)+\mathcal{O}(\Delta t^{2}),
\end{align*}
and the centered spatial difference gives
\[
\frac{H\!\left(t_{n+\frac12},z_{j+\frac12}\right)-H\!\left(t_{n+\frac12},z_{j-\frac12}\right)}{\Delta z}
=H_z(t_{n+\frac12},z_j)+\mathcal{O}(\Delta z^{2}).
\]
Substituting these into the definition of $\xi^{n+\frac12}_{E,j}$ yields
\[
\xi^{n+\frac12}_{E,j}
=\varepsilon_0\varepsilon_\infty E_t(t_{n+\frac12},z_j)
+H_z(t_{n+\frac12},z_j)
+\alpha_{0_t}(t_{n+\frac12},z_j)
+\sigma E(t_{n+\frac12},z_j)
+J_s(t_{n+\frac12},z_j)
+\mathcal{O}(\Delta t^{2}+\Delta z^{2}).
\]
Using the continuous Amp\`ere law in the form $\varepsilon_0\varepsilon_\infty E_t+H_z+\alpha_{0_t}+\sigma E+J_s=0$ evaluated at $(t_{n+\frac12},z_j)$ gives $\xi^{n+\frac12}_{E,j}=\mathcal{O}(\Delta t^{2}+\Delta z^{2}).$ Finally, for \eqref{eq: 3.21c} (again at $(t_{n+\frac12},z_j)$), we use
\begin{align*}
   \left\| \frac{\vec{\alpha}(t_{n+1},z_j)-\vec{\alpha}(t_n,z_j)}{\Delta t} - \vec{\alpha}_t(t_{n+\frac12},z_j)\right\|
&=\mathcal{O}(\Delta t^{2}),\\
\left\|\frac{\vec{\alpha}(t_{n+1},z_j)+\vec{\alpha}(t_n,z_j)}{2}-\vec{\alpha}(t_{n+\frac12},z_j)\right\|
&=\mathcal{O}(\Delta t^{2}),
\end{align*}
and 
\[
    \frac{E(t_{n+1},z_j)+E(t_n,z_j)}{2}=E(t_{n+\frac12},z_j)+\mathcal{O}(\Delta t^{2}).
\]
Moreover, by the same argument as in \cite{oguadimma2025analysis}, one has
\[
\left(\frac{E(t_{n+1},z_j)+E(t_n,z_j)}{2}\right)^3
=E(t_{n+\frac12},z_j)^3+\mathcal{O}(\Delta t^{2}).
\]
Substituting into $\vec{\xi}^{\,\,n+\frac12}_{\alpha,j}$ gives
\[
\left\|\vec{\xi}^{\,\,n+\frac12}_{\alpha,j}-
A\vec{\alpha}_t(t_{n+\frac12},z_j)
+\vec{\alpha}(t_{n+\frac12},z_j)
-\varepsilon_0\varepsilon_d\Big(E(t_{n+\frac12},z_j)+\beta E(t_{n+\frac12},z_j)^3\Big)\hat e_1 \right\|
=\mathcal{O}(\Delta t^{2}).
\]
Using \eqref{4.11c}, $A\vec{\alpha}_t+\vec{\alpha}=\varepsilon_0\varepsilon_d\big(E+\beta E^3\big)\hat e_1$, evaluated at $(t_{n+\frac12},z_j)$ yields $\|\vec{\xi}^{\,\,n+\frac12}_{\alpha,j}\|=\mathcal{O}(\Delta t^{2})$. Combining the three estimates gives
\[
\max_{n,j}\Big\{
\big|\xi^{n}_{H,j+\frac12}\big|,
\big|\xi^{n+\frac12}_{E,j}\big|,
\big\|\vec{\xi}^{\,\,n+\frac12}_{\alpha,j}\big\|
\Big\}
= \mathcal{O}(\Delta z^{2}+\Delta t^{2}),
\]
which completes the proof.
\end{proof}
\begin{remark}
Although \cite{bela2010generalized} considers a different model, our numerical experiments in later sections indicate that analogous stability requirements also apply to the scheme \eqref{eq: 3.21}, namely $\tau_m>0$,\, $\varepsilon_d\ge 0$, \, $\nu < \dfrac{c\,\Delta t}{\Delta x}$.
\end{remark}

\section{Numerical Simulation}\label{sec:results}
In this section, we present numerical experiments that (i) validate the accuracy of the proposed discretizations for the one-dimensional Maxwell--Nonlinear Debye model developed in Section~\ref{LF} and (ii) illustrate the qualitative behavior of wave propagation in nonlinear and random media. 
\subsection{Accuracy Test: A Manufactured Solution}
A manufactured solution is used to demonstrate the accuracy of the scheme. Consider the 1D intrusive nonlinearly forced Maxwell-PC Debye system,
\begin{subequations}\label{eq: 4.7pce}
\begin{align}
H_t &= a_0\,E_z, \\
E_t &= a_1\,H_z + a_2\,\alpha_{0_t}, \\ 
A\,\dot{\vec{\alpha}} + \vec{\alpha} &= \varepsilon_0\varepsilon_d\big(E+\beta E^3\big)\,\hat e_1 + \vec{F},
\end{align}
\end{subequations}
on $\Omega=[0,\pi]$ with $\tau \sim \mathcal{U}(\tau_{a}, \tau_b)$, where $\hat e_1=(1,0,\ldots,0)^T\in\mathbb{R}^{Q+1}$ and $A\in\mathbb{R}^{(Q+1)\times(Q+1)}$.The exact solutions are $E(z,t) = \cos(t)\sin(z), \, H(z,t) = a_0\cos(z)\sin(t), \, \alpha_0(z,t) = -c\cos(t)\sin(z), \,
\alpha_\ell(z,t)=\cos(t)\sin(z),\,\ell=1,\ldots,Q$, with $a_0 = -1/\mu_0, a_1 = a_2= -1/\varepsilon_0 \varepsilon_\infty, c = \varepsilon_0\epsilon_\infty - 1/\mu_0$. The vector $\vec{F}$ is the forcing term to make \eqref{eq: 4.7pce} hold. Periodic boundary conditions are considered.
\begin{table}[H]
\vskip1ex
\caption{Convergence of the Leap-Frog scheme in the $L^2$ norm as $h$ is halved.}
\label{tab:lf_conv_L2_alpha}
\centering
\small
\setlength{\tabcolsep}{4pt}
\resizebox{\textwidth}{!}{%
\begin{tabular}{cccccccccc}
\toprule
$N$ & $h$ &
$\|E-E_h\|_{E}$ & rate &
$\|H-H_h\|_{H}$ & rate &
$\|\alpha_0-\alpha_{0,h}\|_{E}$ & rate &
$\|\vec{\alpha}-\vec{\alpha}_h\|_{\alpha}$ & rate \\
\midrule
50  & 0.02000 & $4.616\times10^{-4}$ & --    & $5.977\times10^{-4}$ & --    & $3.353\times10^{-4}$ & --    & $4.366\times10^{-4}$ & --    \\
100 & 0.01000 & $1.139\times10^{-4}$ & 2.019 & $1.469\times10^{-4}$ & 2.024 & $8.221\times10^{-5}$ & 2.028 & $1.072\times10^{-4}$ & 2.026 \\
200 & 0.00500 & $2.827\times10^{-5}$ & 2.010 & $3.642\times10^{-5}$ & 2.012 & $2.035\times10^{-5}$ & 2.014 & $2.657\times10^{-5}$ & 2.013 \\
400 & 0.00250 & $7.043\times10^{-6}$ & 2.005 & $9.067\times10^{-6}$ & 2.006 & $5.064\times10^{-6}$ & 2.007 & $6.613\times10^{-6}$ & 2.006 \\
800 & 0.00125 & $1.745\times10^{-6}$ & 2.013 & $2.262\times10^{-6}$ & 2.003 & $1.255\times10^{-6}$ & 2.012 & $1.641\times10^{-6}$ & 2.010 \\
\bottomrule
\end{tabular}%
}
\end{table}
Tables~\ref{tab:lf_conv_L2_alpha} reports the discrete $L^2$ and $L^\infty$ errors for $E(z,t), H(z,t)$ and the first few modes ($Q = 2$) of the PCE system, under successive mesh refinements. As the grid spacing is halved, the errors decrease by approximately a factor of four, andAcknowledgments the measured convergence rates remain close to two in both norms. Similar result is observed for higher values of $Q$. This confirms that the proposed Yee discretization achieves the expected second-order accuracy in space and time for the manufactured-solution test. We note that for $Q = 0$, the system corresponds to the deterministic model.

\subsubsection{Convergence in $Q$}
To validate convergence with respect to the chaos truncation order, we drive the system using a left-boundary hard source for the electric field,
\[
E(0,t)=\sin(2\pi f_0 t)\exp\!\left(-\frac{(t-t_0)^2}{s^2}\right),
\]
i.e., a Gaussian-modulated sinusoid. In the computations, we set $f_0=5$, $t_0=6/f_0$, and $s=1.5/f_0$. All fields are initialized to zero, and the right boundary is taken as perfectly conducting. For each order of truncation $Q$, we evolve the system in time and compute the PCE coefficients $\{\alpha_k(x,t)\}_{k=0}^{Q}$ throughout the simulation. The mean and variance are then formed at each time level,
\[
\mathbb{E}[P(x,t)]=\alpha_0(x,t),\qquad 
\mathrm{Var}[P(x,t)]=\sum_{k=1}^{Q}\frac{\alpha_k(x,t)^2}{2k+1},
\]
and the corresponding errors are reported as space--time integrated quantities over the $(x,t)$ domain.
\begin{figure}[H]
  \centering
  \begin{subfigure}{0.48\linewidth}
    \centering
    \includegraphics[width=\linewidth]{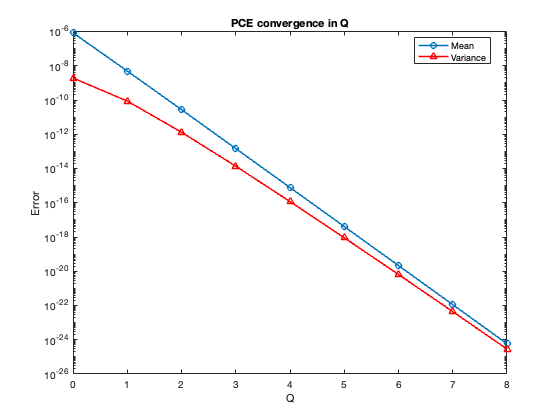}
    \caption{$\tau_r = 0.5\tau_m$}
  \end{subfigure}\hfill
  \begin{subfigure}{0.48\linewidth}
    \centering
    \includegraphics[width=\linewidth]{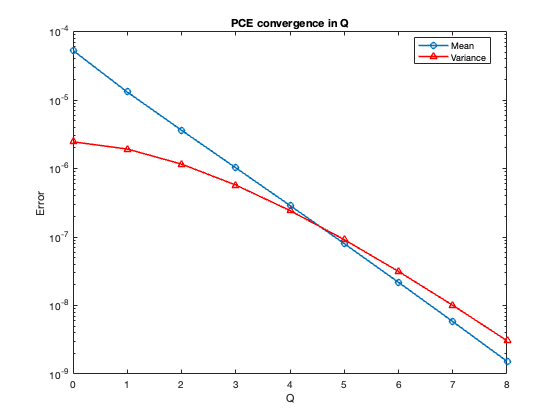}
    \caption{$\tau_r = 0.95\tau_m$}
  \end{subfigure}
  \caption{Error over spatial and temporal domain with Uniform random variable by Legendre Chaos.}
  \label{fig:pceconv}
\end{figure}
Figure~\ref{fig:pceconv} shows a monotone decay of the mean error, $\|\alpha_0 - \alpha_{0_{\text{ref}}}\|_{\alpha}$, as the chaos order $Q$ increases, with an approximately linear trend on the semi-log scale in both panels, consistent with exponential convergence. For the moderate uncertainty case $\tau_r=0.5\,\tau_m$ (left), the variance error exhibits a similar clean, near-linear decay. In contrast, when the randomness is increased to $\tau_r=0.95\,\tau_m$ (right), the variance curve departs from this simple linear trend and decays more irregularly at low-to-moderate orders. This behavior suggests that stronger parameter uncertainty excites higher-order stochastic content in the solution, so that more chaos modes are required before the variance converges in a smooth, asymptotically exponential manner. Nevertheless, as $Q$ increases, the variance error continues to decrease and approaches the same error floor observed for the mean.

\subsection{Time Domain Propagation}
We use the following parameters (assuming the dielectric parameters of water) \cite{banks2006electromagnetic}: 
\begin{align}\label{eq: 4.8}
    \epsilon_0 = 8.854 \times 10^{-12}, \; \epsilon_\infty &= 5.5, \; \epsilon_s = 80.1, \; \tau_m = 8.1 \times 10^{-12} \notag \\
    c = 2.9980 &\times 10^8,  \; \mu_0 = 1.2566 \times 10^{-6}, \; \sigma = 10^{-5}.
\end{align} 
We perform a series of numerical experiments with a windowed input signal (generated by an antenna at $z = \ell/2$):  
\begin{equation}
    J_s = |E|^2 \sin^3(\omega t) \chi_{[0, t_f]}(t) \delta (z - \ell/2)
\end{equation}
with $t_f = 5/\omega$, $\omega=2 \pi f$, and $f = 1, 10, 100$ GHz and 1 THz. The quantity $|E|^2$ is considered as a constant intensity of the source.  
We use absorbing boundary conditions at the end of the slab to prevent the reflection of waves 
\begin{align}
    \frac{\sqrt{\epsilon_s}}{c} \frac{\partial E}{\partial t} - \frac{\partial E}{\partial z} \bigg|_{z=0} &= 0, \; \; \; \; 
    \frac{\sqrt{\epsilon_s}}{c} \frac{\partial E}{\partial t} + \frac{\partial E}{\partial z} \bigg|_{z=\ell} = 0. 
\end{align} 
We complete the system with initial conditions 
\begin{align}
    E(0,z) = 0, \; H(0,z) = 0, \; P(0,z) = 0, \; J_s(0,z) = 0 \hspace{3mm} \forall \; z \in (0,\ell).
\end{align} 
The spatial resolution of the grid is \( \Delta z = \lambda / 400 \), $\lambda = c/f$, and the time step satisfies \( \Delta t = (\nu \Delta z) / c \) for a Courant number \( \nu = 0.5 \).  In the following, we present case studies for two frequencies: 100 GHz and 1 THz using parameter values $|E|^2 $= 100 V/m,\; $\beta = 5 \times 10^{-6}$,\; $\tau_r = 0.95\tau_m,\; Q=5$.

Our time-domain simulations in Fig.~\ref{fig: fig28}--\ref{fig: fig35} show that the electric field propagates slower and experiences greater attenuation in the nonlinear models compared to the linear models across all frequencies. At the lowest frequency ($f \leq 10$ GHz),  the electric field values are similar in both the deterministic and random models. However, as the frequency increases, the electric field propagates faster in the deterministic model while experiencing greater attenuation in the random model. The analysis and numerical results of propagation in deterministic model presented in this report are consistent with prior theoretical treatments and experimentally observed trends reported in \cite{pleshko1969precursors,jackson1998classical,albanese1993ultrashort, HTBanksGabriella}.

\begin{figure}[h]
    \centering
    \includegraphics[scale = 0.22]{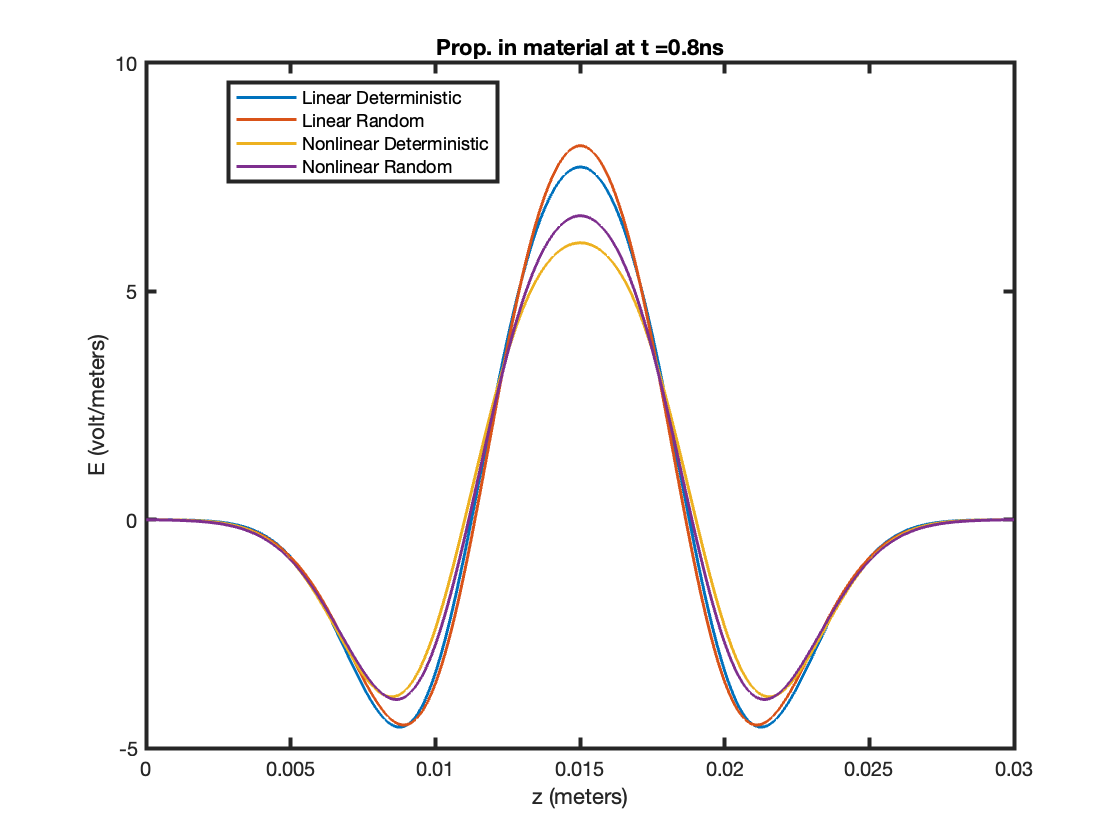}
    \caption{Snapshots of the Electric Field in the Material at $t = 0.8$ns}
    \label{fig: fig28}
\end{figure}

\begin{figure}[h]
    \centering
    \includegraphics[scale = 0.22]{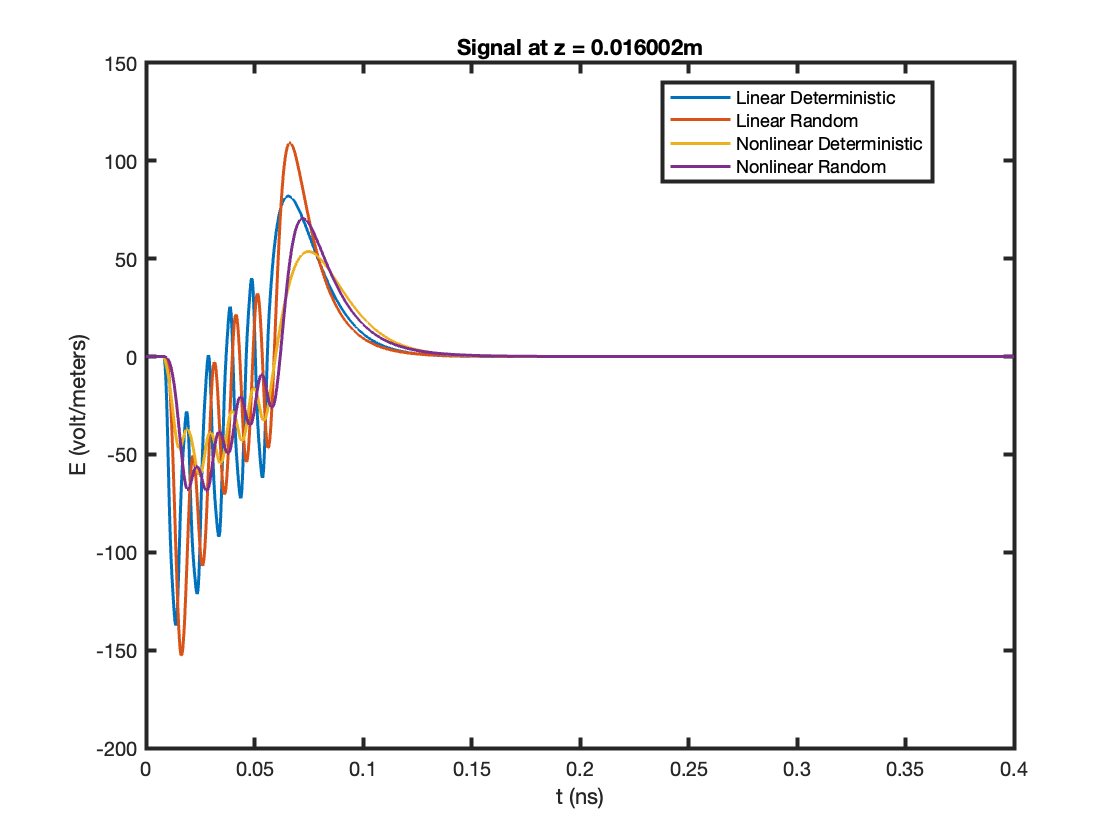}
    \includegraphics[scale = 0.22]{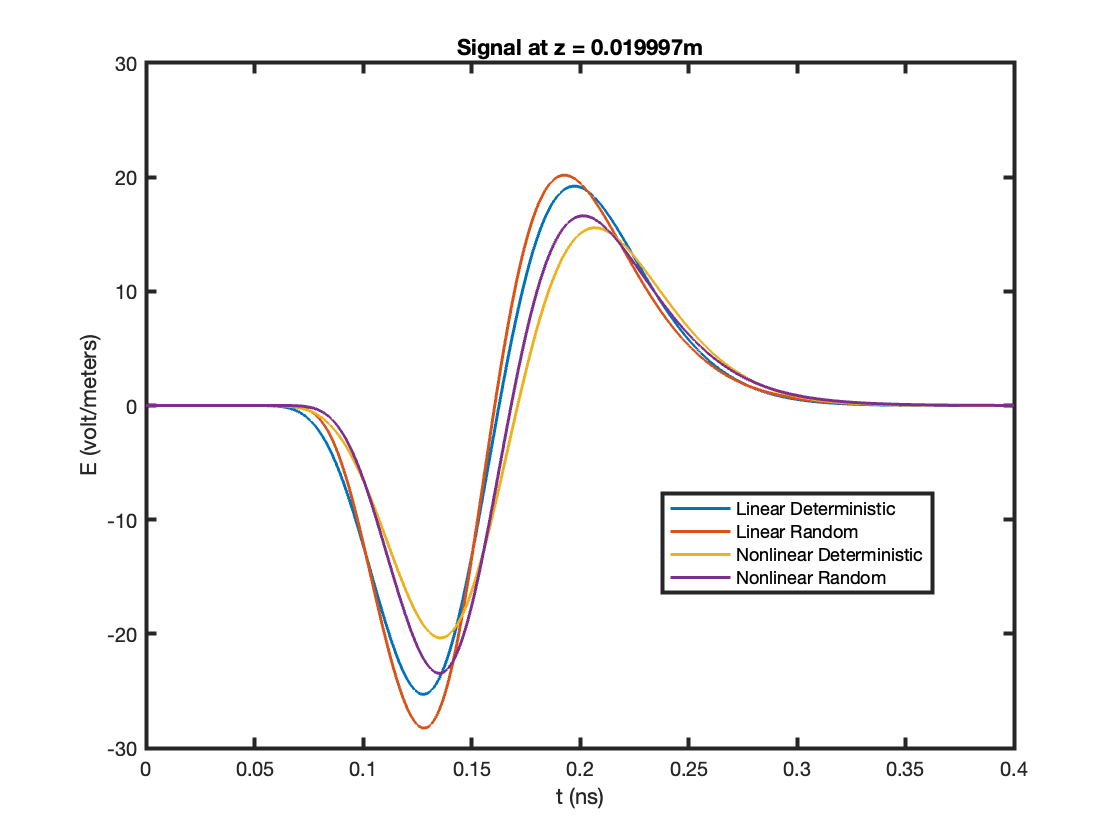}
    \caption{The electric field recorded at a distance of 0.016m and 0.02m from the antenna}
    \label{fig: fig29}
\end{figure}

\begin{figure}[h]
    \centering
    \includegraphics[scale = 0.2]{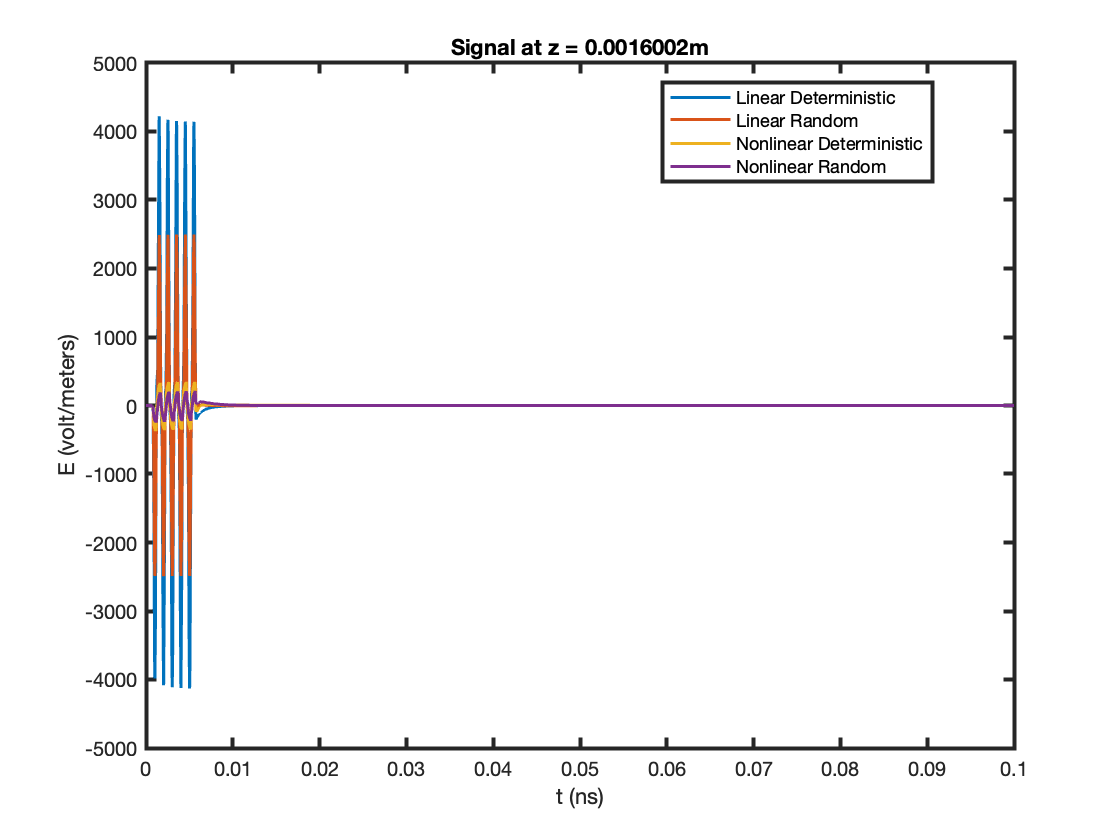}
    \includegraphics[scale = 0.2]{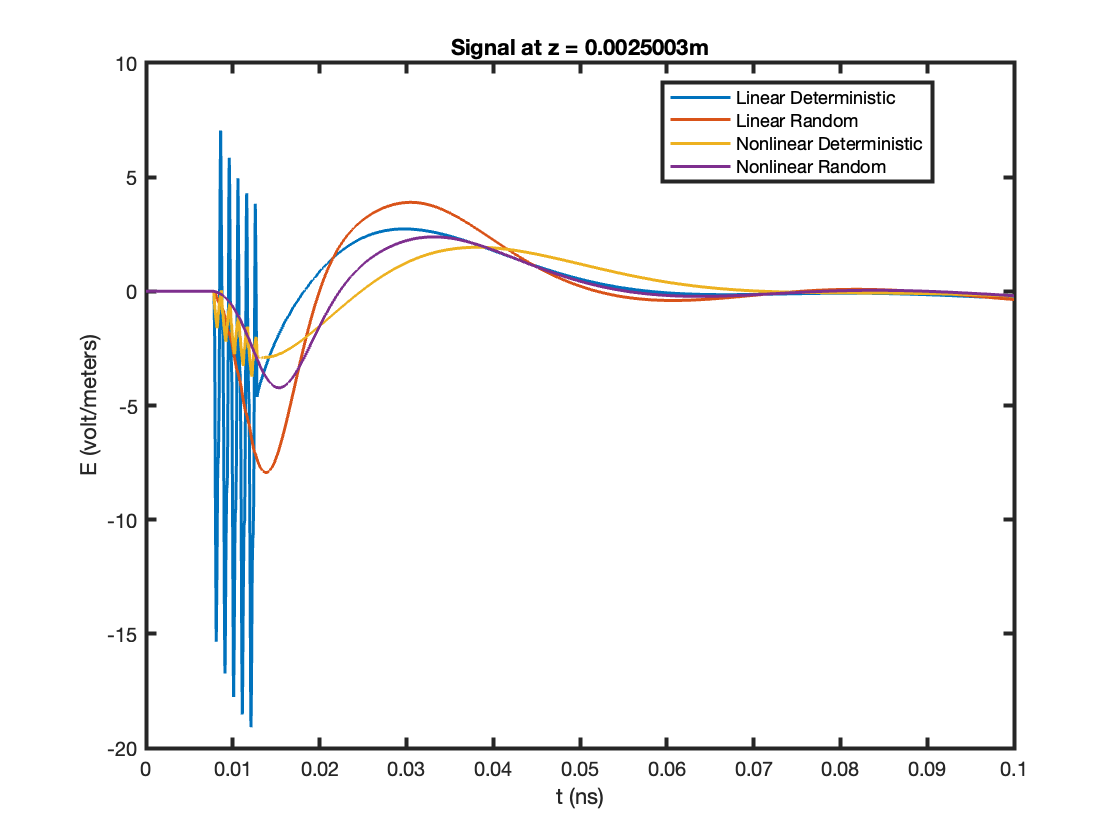}
    \caption{The electric field recorded at a distance of 0.0016m and 0.0025m from the antenna}
    \label{fig: fig35}
\end{figure}

\subsection{Fourier Transforms of Simulations} 
We define the power spectrum of the time-domain signal $E(\hat{z},t)$ at fixed \(\hat z\) by
\begin{align}
    P(\hat z,\omega)\;=\;\bigl|\widehat{E}(\hat z,\omega)\bigr|^2,
\end{align} 
where $\widehat{E}(\hat z,\omega)$ is the frequency-domain representation of the signal. We first record the electric field at a fixed receiver position throughout the simulation, then zero-pad the resulting discrete time series to the next power of two in order to enhance frequency resolution and computational efficiency for the FFT. After applying the FFT and normalizing by the total number of time steps, we extract the single-sided amplitude spectrum. The power spectrum was then computed as the squared magnitude of the single-sided FFT, yielding the distribution of signal power with respect to frequency. In the following, we present case studies for two frequencies: 100 GHz and 1 THz.

\ 

\noindent 
\textbf{Frequency: 100 GHz, $|E|^2 $: 100 V/m,\;$\tau_r = 0.95\tau_m,\; Q=5$}

\begin{figure}[h]
    \centering
    \includegraphics[scale = 0.2]{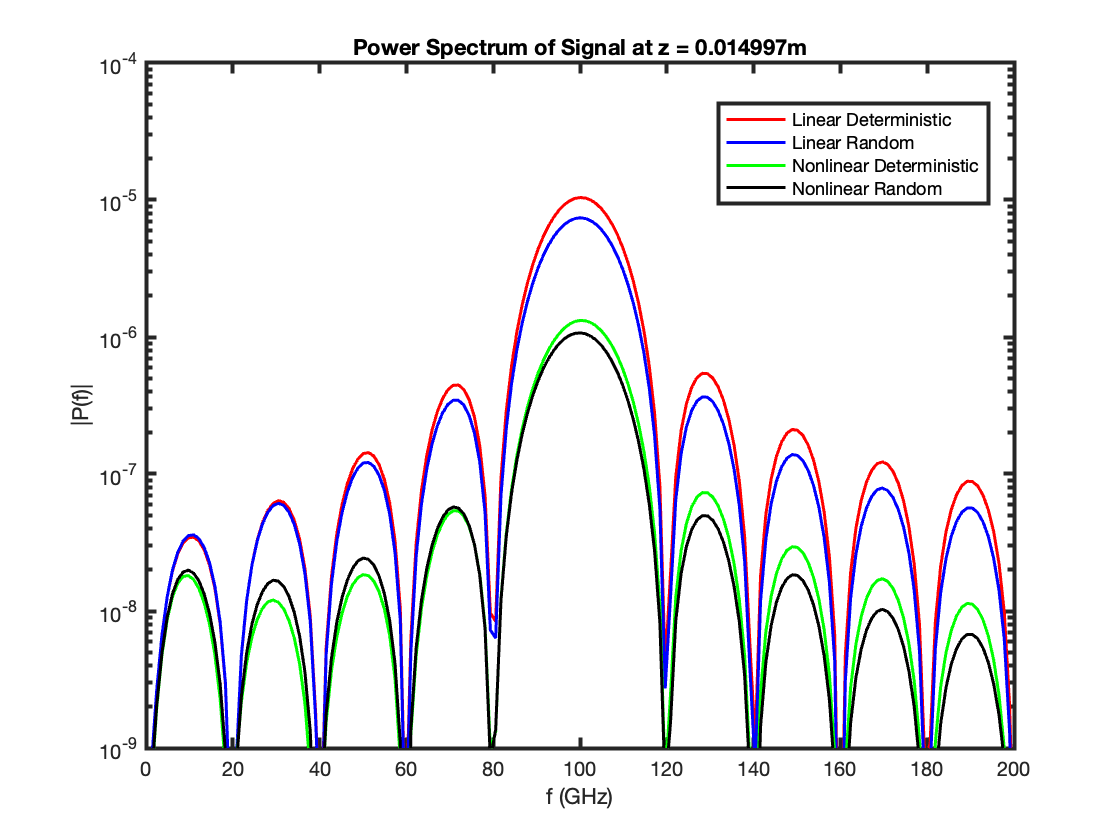}
    \includegraphics[scale = 0.2]{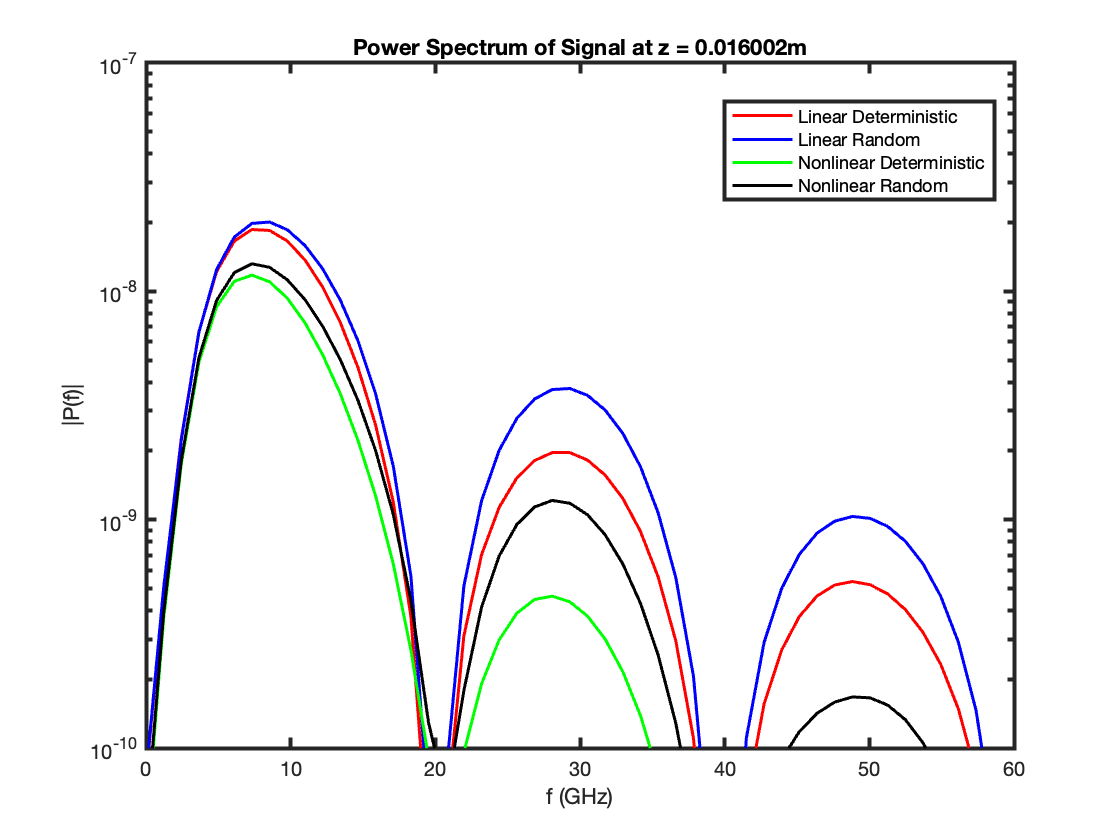}
    \caption{Power Spectrum of Signal at the antenna and at a distances of 0.015m and 0.016m from the antenna}
    \label{fig: fig37}
\end{figure}
\noindent 
\textbf{Frequency: 1 THz, $|E|^2 $: 100 V/m,\;$\tau_r = 0.95\tau_m,\; Q=5$}
\begin{figure}[h]
    \centering
    \includegraphics[scale = 0.2]{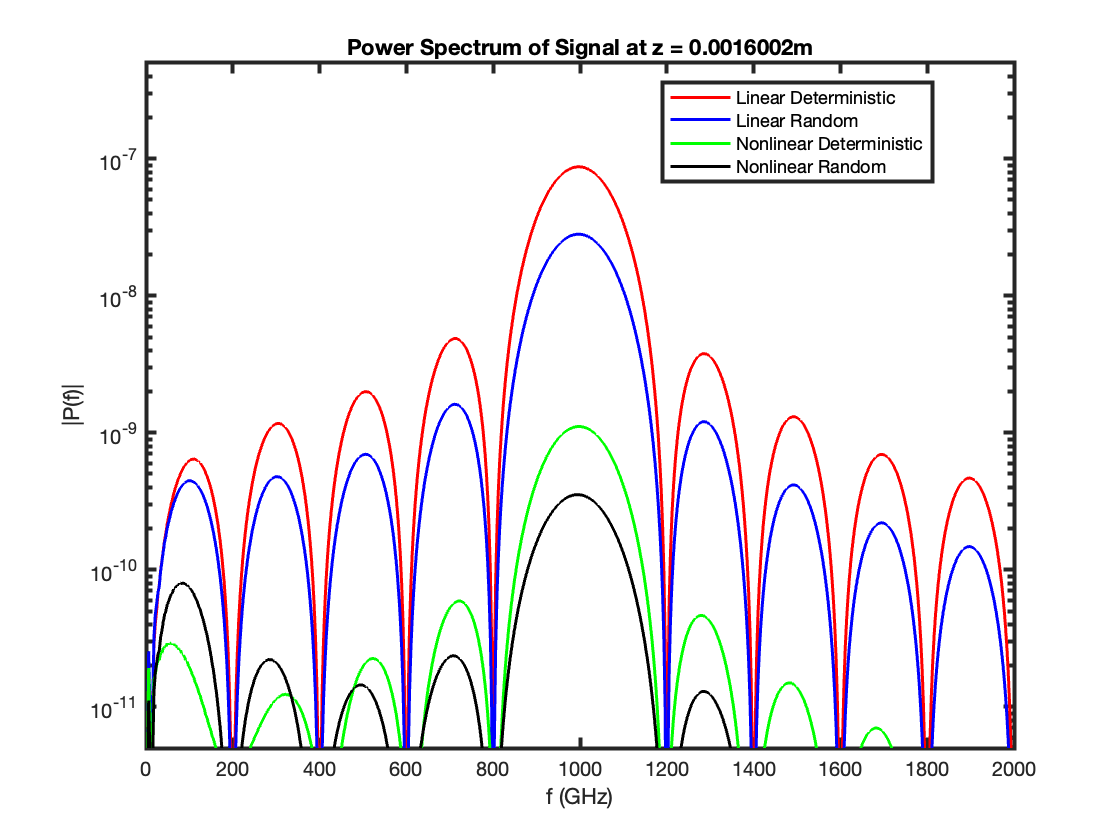}
    \includegraphics[scale = 0.2]{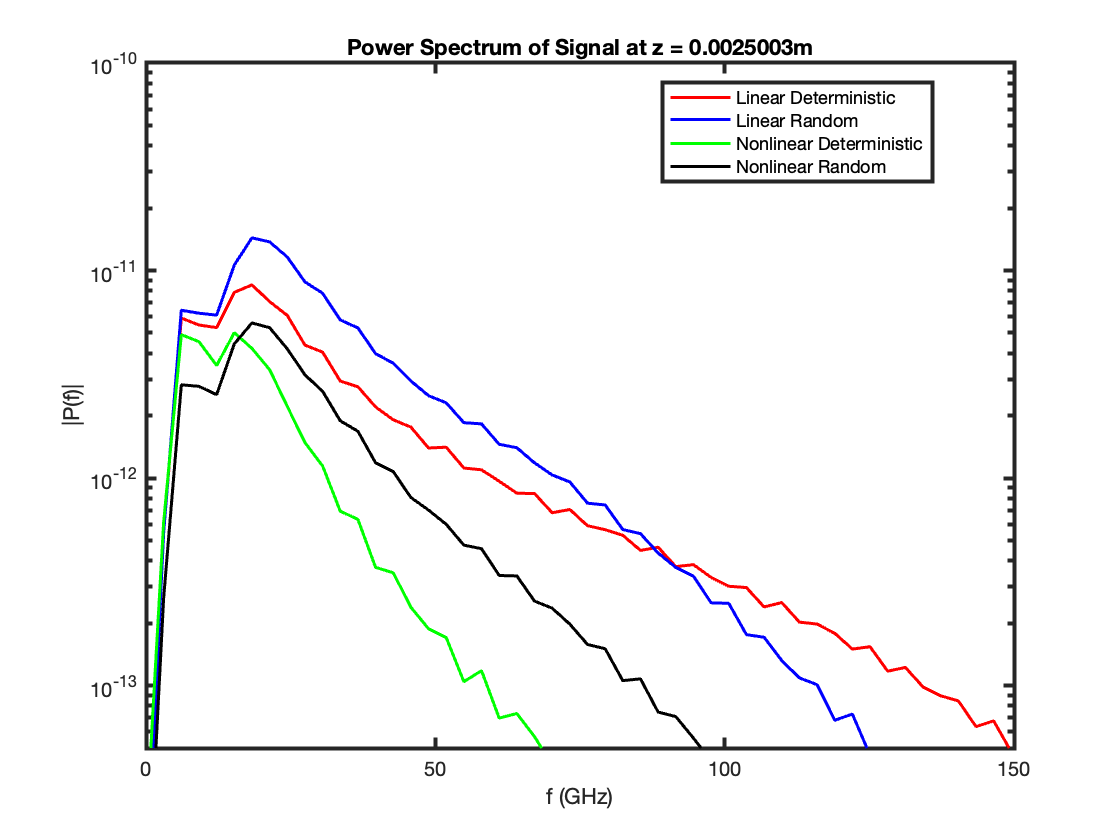}
    \caption{Power Spectrum of Signal at a distance of 0.0015m and 0.0025m from the antenna}
    \label{fig: fig38}
\end{figure}

For each frequency in Fig.~\ref{fig: fig37}--\ref{fig: fig38}, the plots on the left show that the linear model exhibits less attenuation than the more realistic nonlinear model. However, in both cases, the signal in the deterministic model is initially less attenuating, but switches to more attenuating deeper into the material. This clearly shows the complex interaction between random and nonlinear effects in time domain simulations.

\section{Inverse Problem}\label{sec:parameterIdentif} 
We consider an example material design problem, where we seek to determine the best set of physical parameters to achieve optimal wave behavior in the given medium. We define the design parameter set as 
\begin{align}
     \mathbf{q} = [\epsilon_\infty, \epsilon_d, \tau_r, \tau_m, \sigma, \beta]
\end{align}
The goal is to find the optimal parameters $\mathbf{q}^*$ that minimize the objective function 
\begin{align}
    \mathcal{T}(\mathbf{q}) = - \max_t|E(t, \hat{z}; \mathbf{q})| + \alpha \mathcal{J}_{H^1}(\mathbf{q}), 
\end{align}
where $|E(t, \hat{z}; \mathbf{q})|$ represents the electric field amplitude at a fixed distance $\hat{z}$ and $\mathcal{J}_{H^1}(\mathbf{q})$ is a penalty function defined as 
\begin{align}
     \mathcal{J}_{H^1}(\mathbf{q}) = \left(\int_t |E'(t, \hat{z};\mathbf{q})|^2 dt\right)^{1/2},
\end{align}
with $\alpha$ as a regularization parameter. 

\subsection{Numerical Results: Significance Test}
We use the parameters in \eqref{eq: 4.8} as the initial guess for the vector $\mathbf{q} = [\epsilon_\infty, \epsilon_d, \tau_r, \tau_m, \sigma, \beta]$
along with upper and lower bounds given by 
\begin{equation}
    \mathbf{q}_{\text{ub}} = [10,100,1,10^{-10}, 10^{-3}, 10^{-2}],\quad  \mathbf{q}_{\text{lb}} = [1,70,0,10^{-14}, 10^{-6}, 10^{-8}].
\end{equation}
We employ MATLAB’s \texttt{fmincon} function, using the interior‐point algorithm to solve the constrained optimization problem. 

We consider the impact of including randomness and nonlinearity in parameter identification for a dielectric material model. Our analysis focused on the utility of including two specific parameters: $\tau_r$ and $\beta$ for optimization purposes. The parameter $\tau_r$ represents a relaxation time variable whose inclusion introduces randomness into the model formulation. Conversely, the parameter $\beta$ governs a nonlinear term and thus introduces nonlinearity. Numerical experiments are conducted at two frequency levels (1 GHz and 1 THz) and for multiple $\alpha$ values to simulate different regularization strengths. The impact of $\tau_r$ and $\beta$ is assessed by comparing the final cost values between models where these parameters are included or excluded. A relative modeling error is computed to quantify the significance of the inclusion of each parameter.

\subsubsection{Inclusion of Randomness ($\tau_r$)}

\textbf{Frequency: 1 GHz. $\alpha = 10^{-11}$, $\tau_m = 8.1 \times 10^{-12}$ }

The following results show that the random model achieves a lower final cost, indicating that it better satisfies the optimization objective. 

        \begin{center}
        \begin{tabular}{ccc}
           \textbf{Model} & \textbf{Initial Cost}  & \textbf{Final Cost}\\
           \midrule 
           Deterministic &  $ 2.295439$ &  $-5.972467$ \\
           Random & $ 2.442979$ &  $-6.021237$ \\
        \end{tabular}            
        \end{center}

In practical terms, this suggests that including $\tau_r$ helps improve the overall optimization performance. To quantify the impact of $\tau_r$ on the optimization outcome, we compare the deterministic and random formulations. Let $\mathcal{T}(\tau_r=0)$ denote the final cost from the deterministic model and $\mathcal{T}(\tau_r \neq 0)$ the final cost from the random model. We define the relative error:
\begin{align}\label{relerr}
    e =\; \frac{|\mathcal{T}(\tau_r=0) \;-\; \mathcal{T}(\tau_r \neq 0)|}{\bigl|\mathcal{T}(\tau_r \neq 0)\bigr|}.
\end{align}
Using the values
\[
\mathcal{T}(\tau_r=0) = -5.972467, 
\quad 
\mathcal{T}(\tau_r \neq 0) = -6.021237,
\]
we obtain a relative error
\[
e =\;
\frac{|(-5.972467) - (-6.021237)|}{\lvert -6.021237\rvert}
\;=\;
\frac{0.04877}{6.021237}
\;\approx\;
0.0081
\] 

In Fig.~\ref{fig: fig42}, we plot the optimized electric-field responses for both the deterministic model ($\tau_r=0$) and the random model ($\tau_r\neq 0$). The two optimized signals are nearly indistinguishable, indicating that introducing relaxation-time uncertainty has little visible impact on the optimized waveform. Moreover, the reduction in the final cost due to the inclusion $\tau_r$ is marginal, corresponding to relative error of 0.81\%. This indicates that randomness has negligible effect on optimization in this setting.
\begin{figure}[h]
    \centering
     \includegraphics[scale = 0.2]{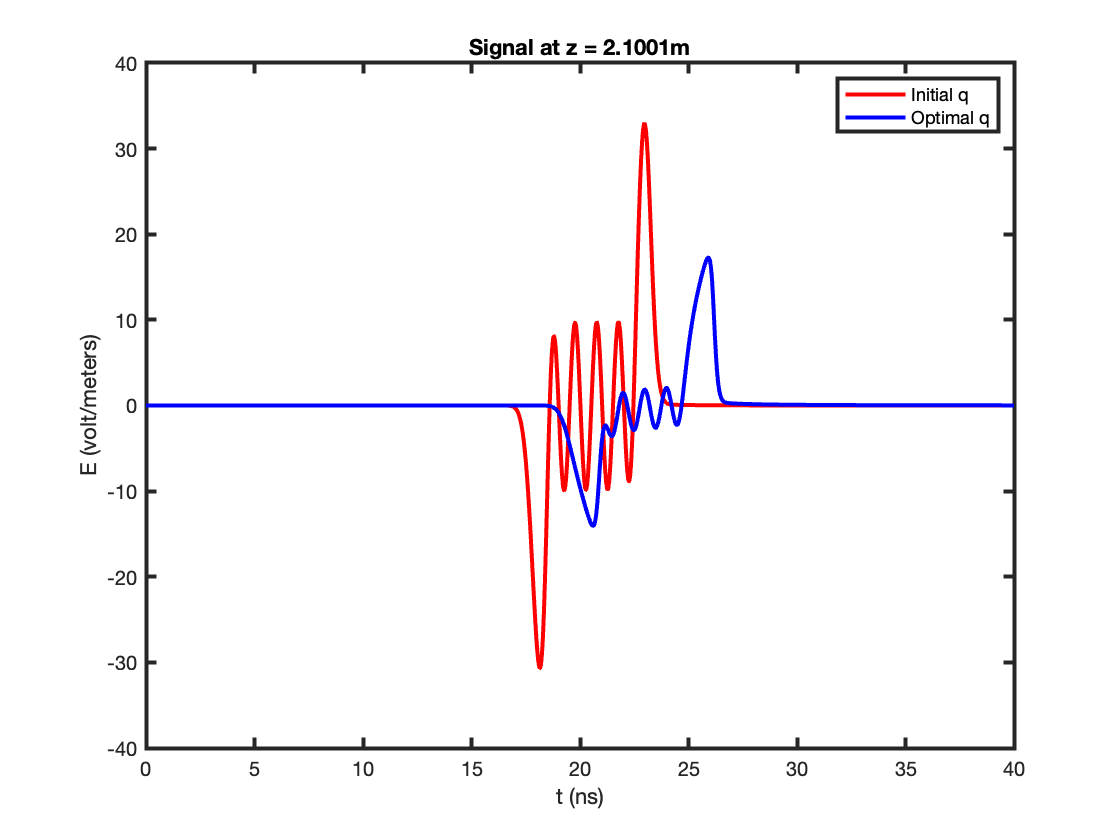}
     \includegraphics[scale = 0.2]{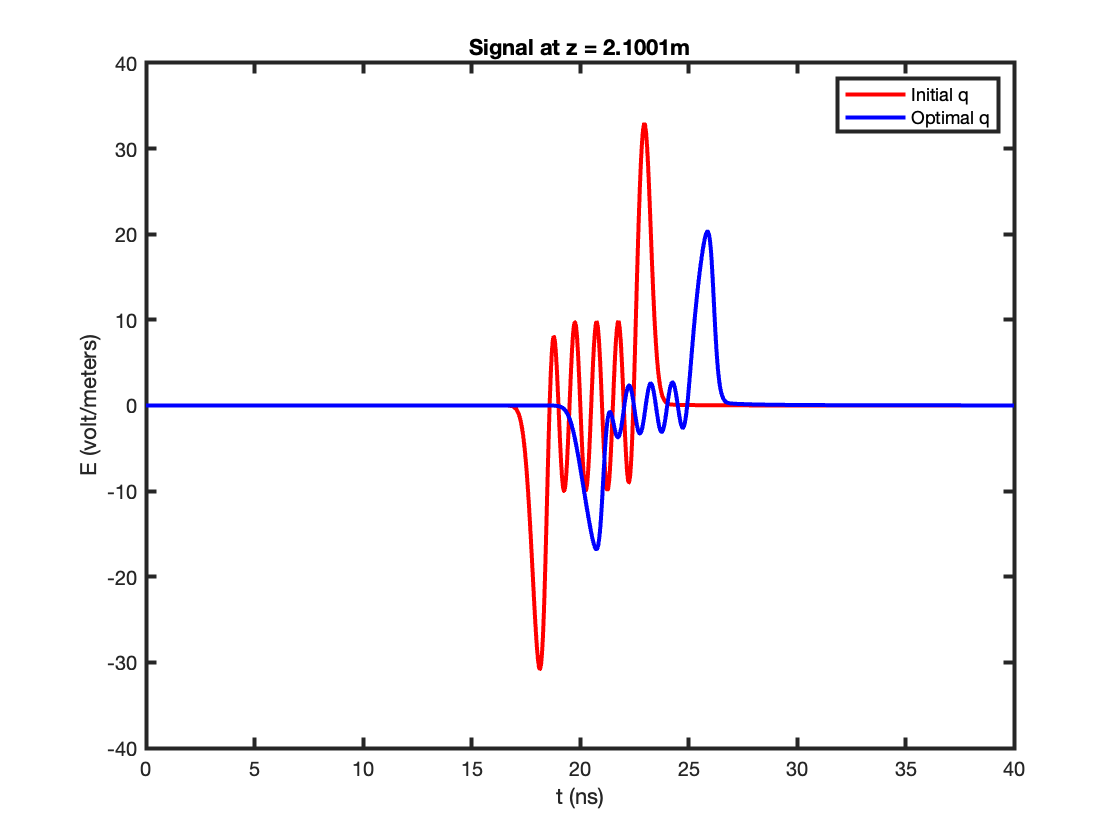}
    \caption{Optimization results for deterministic (left) and random (right) models}
    \label{fig: fig42}
\end{figure}

\subsubsection{Inclusion of Nonlinearity ($\beta$)}
\textbf{Frequency: 1 THz. $\alpha = 10^{-12}$, $\tau_m = 8.1 \times 10^{-12}$ }

For this case, we test the significance of including nonlinearity. The optimization results

        \begin{center}
        \begin{tabular}{ccc}
           \textbf{Model} & \textbf{Initial Cost}  & \textbf{Final Cost}\\
           \midrule 
           $\beta = 0$ &  $14.10884$ &  $4.857213$ \\
           $\beta \neq 0$ &  $  12.27078$ &  $2.762999$ \\
        \end{tabular}
        \end{center}
\noindent 
yield a relative error of $e = 0.758$.
The optimal plots in both panels of Fig.~\ref{fig: fig45} are somewhat different. Moreover, the relative error of $75.8\%$ indicates that, in this high frequency regime, neglecting the nonlinear coefficient $\beta$ severely degrades the optimization performance. Therefore, incorporating nonlinearity is essential to obtain effective and physically meaningful optimal solutions in this regime.
\begin{figure}[H]
    \centering
     \includegraphics[scale = 0.2]{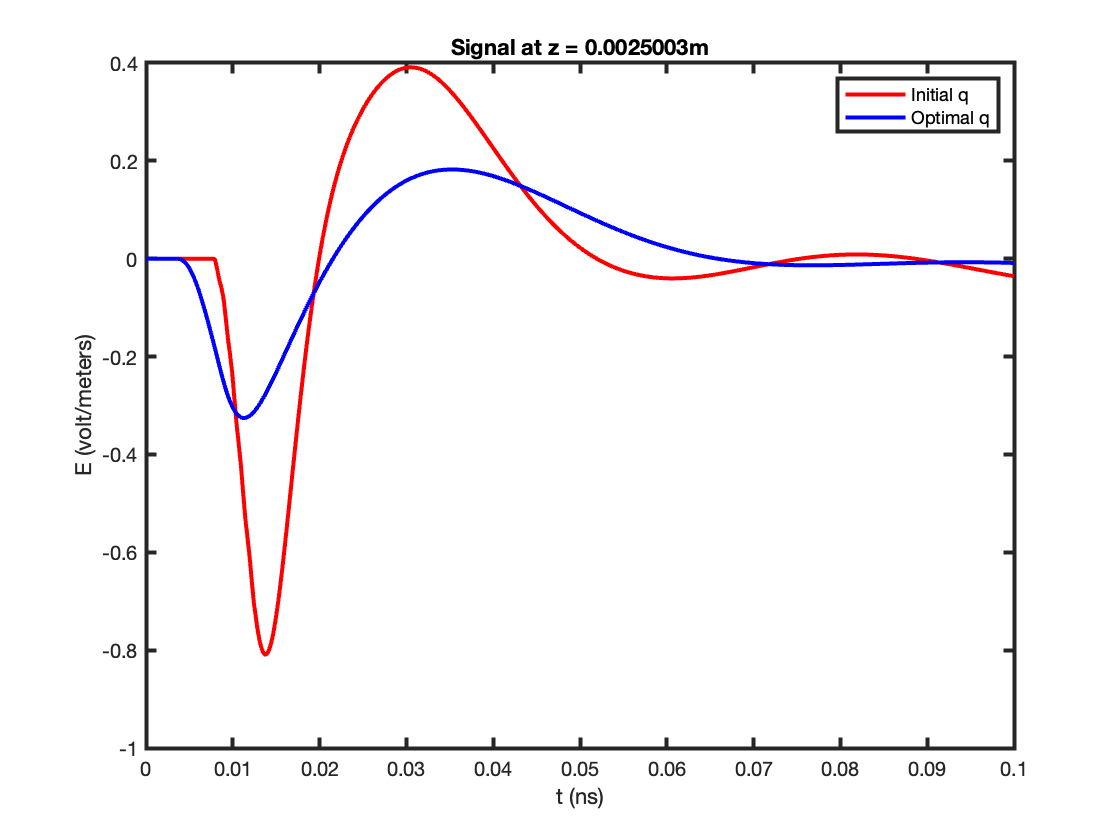}
     \includegraphics[scale = 0.2]{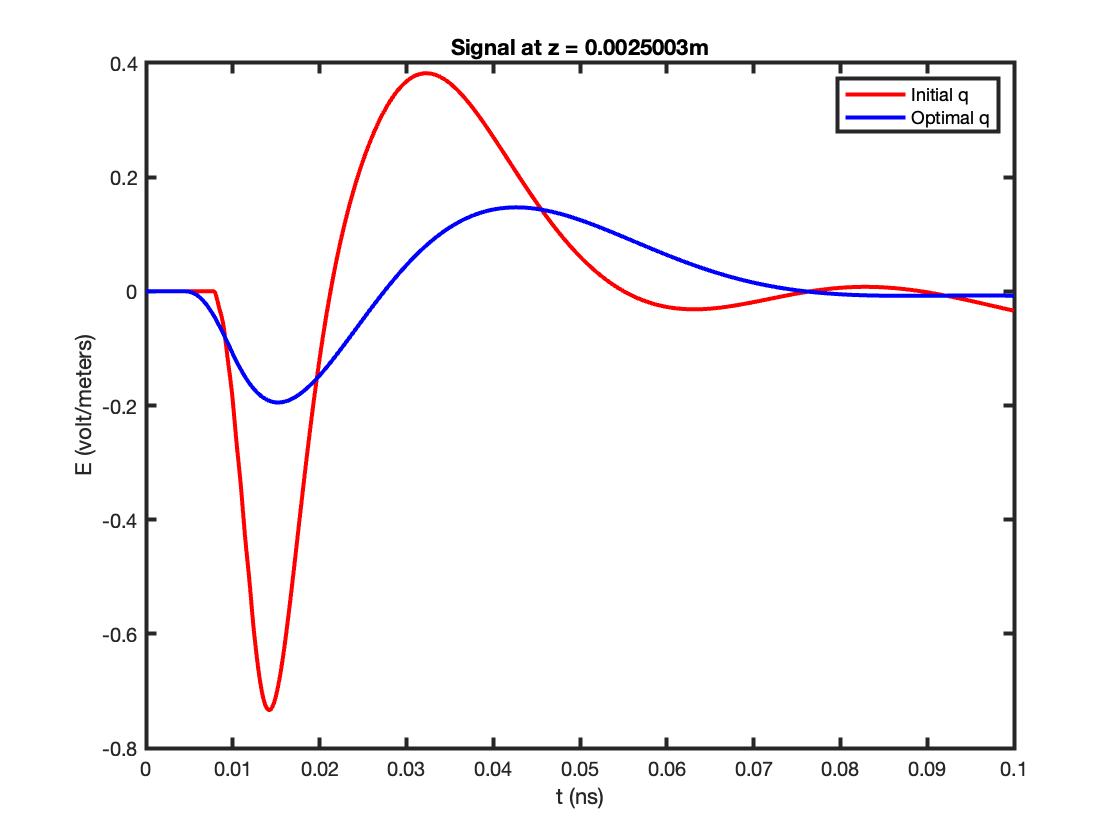}
    \caption{Optimization results for linear (left) and nonlinear (right) models}
    \label{fig: fig45}
\end{figure}

\subsubsection{Inclusion of Randomness and Nonlinearity ($\tau_r$ and $\beta$)}

\textbf{Frequency: 1 GHz. $\alpha = 10^{-11}$, $\tau_m = 8.1 \times 10^{-12}$ }

For this case, we test the significance of including both randomness and nonlinearity. The optimization results

        \begin{center}
        \begin{tabular}{ccc}
           \textbf{Model} & \textbf{Initial Cost}  & \textbf{Final Cost}\\
           \midrule 
           $\beta,\tau_r = 0$ &   $ 4.971905$ &  $  -2.162200$ \\
           $\beta,\tau_r \neq 0$ & $ 2.442979$ &  $-6.021237$ \\[2.0ex]
        \end{tabular}
        \end{center}
\noindent 
yield a relative error $e = 0.641$.
The optimal solutions in Fig.~\ref{fig: fig46} differ when $\tau_r$ and $\beta$ are both excluded from the optimization process. These results suggest that although $\tau_r$ may not be particularly significant on its own for low frequencies, it becomes important when combined with $\beta$.  
\begin{figure}[H]
    \centering
    \includegraphics[scale = 0.2]{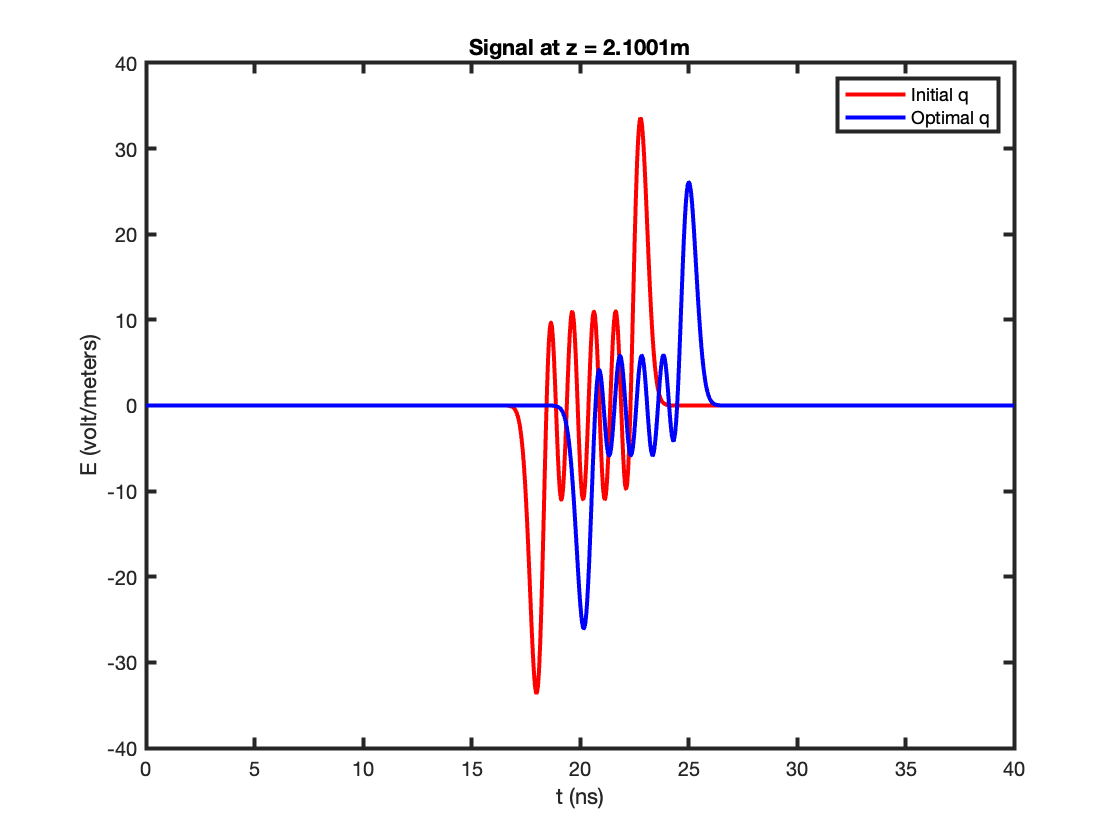}
     \includegraphics[scale = 0.2]{3ai.png}
    \caption{Optimization results for linear deterministic (left) and nonlinear random (right) models}
    \label{fig: fig46}
\end{figure}

\section{Conclusions}\label{sec:conclusions}
We have shown that we can improve the traditional nonlinearly forced Debye model in \cite{HTBanksGabriella} by replacing the average relaxation time $\tau$ with a distribution of relaxation times. The Polynomial Chaos method provides us with a convenient means of representing the stochastic polarization, $\mathcal{P}$ as a linear combination of orthogonal polynomials. By projecting into finite random space, we are able to replace a nonlinear random ordinary differential equation with a system of deterministic ODEs. Combining these with Maxwell’s and Faraday’s Law and noting that the electric field $E$ depends only on the macroscopic polarization $\mathbb{E}(P) \approx \alpha_0$, we obtain the polynomial chaos model \eqref{MPC} for an electromagnetic field propagating through a nonlinear dispersive dielectric media. The proposed discretization retains second-order accuracy for the coupled nonlinearly forced Maxwell-PC Debye system, consistent with both the truncation-error analysis and the convergence experiments.

Our results clearly show that the interplay between nonlinearity and randomness has a substantial impact on wave dynamics in Debye materials. These effects are not additive. Nonlinearity and randomness combine to create wave behaviors that are not seen with either alone.  We have shown that inverse problems can be solved for nonlinear random polarization models.  Regularization was essential to suppress unrealistic oscillations and enforce smooth, physically realistic waveforms. The significance test reveals that nonlinearity ($\beta$) is necessary for achieving optimal optimization results under strong regularization and high frequencies, while randomness ($\tau_r$) is generally not needed on its own. However, $\tau_r$ can become significant when combined with $\beta$. 

\section*{Acknowledgments and Author Statement}
This work was supported by NSF under DMS grant number 2012882. The authors report there are no competing interests to declare. The first author contributed the conception and design, revising it critically for intellectual content and the final approval of the version to be published; the second author contributed formal analysis and investigation, and the drafting of the paper. All authors agree to be accountable for all aspects of the work.

\bibliographystyle{plain}
\bibliography{ref}

\end{document}